\makeatletter
\def\input@path{{template/}{sections/}}
\makeatother

\documentclass[sigconf, nonacm]{acmart}

\newcommand\vldbdoi{XX.XX/XXX.XX}
\newcommand\vldbpages{XXX-XXX}
\newcommand\vldbvolume{14}
\newcommand\vldbissue{1}
\newcommand\vldbyear{2020}
\newcommand\vldbauthors{Preston Vander Vos,
    Alberto Sonnino,
    Giorgos Tsimos,
    Philipp Jovanovic,
    Lefteris Kokoris-Kogias}
\newcommand\vldbtitle{Bluebottle: Fast and Robust Blockchains through Subsystem Specialization}
\newcommand\vldbavailabilityurl{https://github.com/phvv/mysticeti/tree/odontoceti}
\newcommand\vldbpagestyle{plain}

\newif\ifpublish
\publishtrue

\newif\iflong
\longtrue

\usepackage{xspace}
\usepackage{cleveref}
\usepackage{enumitem}
\usepackage{algorithm}
\usepackage{booktabs}
\usepackage{algpseudocode}
\usepackage{multicol}
\usepackage{breakcites}
\usepackage{subcaption}
\usepackage{listings}
\usepackage{mathtools}
\usepackage{fancyhdr}

\newtheorem{definition}{Definition}
\newtheorem{lemma}{Lemma}
\newtheorem{corollary}{Corollary}
\newtheorem{theorem}{Theorem}

\newcommand{\case}[1]{\par\noindent\textbf{Case:} #1\par}

\newcommand{\codelink}{
    \ifpublish
    \url{https://github.com/phvv/mysticeti/tree/odontoceti} (commit \texttt{e02aeba})
    \else
        \url{https://anonymous.4open.science/r/bluebottle-core}
    \fi
}
\newcommand{\asynccodelink}{
    \ifpublish
        \url{https://github.com/phvv/mysticeti/tree/odontoceti-async} (commit \texttt{7b503fa})
    \else
        \url{https://anonymous.4open.science/r/bluebottle-core}
    \fi
}
\newcommand{\dashboardlink}{
    \ifpublish
        \url{https://github.com/asonnino/mysticeti/blob/odontoceti/crates/orchestrator/assets/grafana-dashboard.json}
    \else
        \url{https://anonymous.4open.science/r/bluebottle-core/crates/orchestrator/assets/grafana-dashboard.json}
    \fi
}
\newcommand{\longversion}{the long version of the paper~\cite{vandervos2025bluebottle}\xspace}
\newcommand{\extref}[2]{\iflong\Cref{#2}\else the long version of the paper~\cite{vandervos2025bluebottle} (#1)\fi}

\ifpublish
  \newcommand{\alberto}[1]{\textcolor{blue}{\textbf{Alberto:} #1}}
  \newcommand{\lef}[1]{\textcolor{brown}{\textbf{Lef:} #1}}
  \newcommand{\giorgos}[1]{\textcolor{purple}{\textbf{Giorgos:} #1}}
  \newcommand{\philipp}[1]{\textcolor{teal}{\textbf{Philipp:} #1}}
  \newcommand{\preston}[1]{\textcolor{violet}{\textbf{Preston:} #1}}
\else
  \newcommand{\alberto}[1]{}
  \newcommand{\lef}[1]{}
  \newcommand{\giorgos}[1]{}
  \newcommand{\philipp}[1]{}
  \newcommand{\preston}[1]{}
\fi

\newcommand{\fabc}{\ensuremath{f_{AbC}}\xspace}

\newcommand{\sysname}{\textsf{BlueBottle}\xspace}
\newcommand{\lmysticeti}{\textsf{BB-Core}\xspace}
\newcommand{\asynclmysticeti}{\textsf{BB-Core-Async}\xspace}
\newcommand{\gossip}{\textsf{BB-Guard}\xspace}

\newcommand{\para}[1]{\smallskip\noindent\textbf{#1}.}
\ifpublish
  \newcommand{\change}[1]{#1}
\else
  \newcommand{\change}[1]{\begingroup\color{orange}#1\endgroup}
\fi

\newcommand{\altcode}[1]{\textcolor{orange}{#1}}
\newtheorem{observation}{Observation}

\newcommand{\propose}{\textsf{Propose}\xspace}
\newcommand{\decision}{\textsf{Decision}\xspace}
\newcommand{\A}{\text{$\mathcal{A}$}\xspace}

\newcommand{\scommit}{\texttt{commit}\xspace}
\newcommand{\sskip}{\texttt{skip}\xspace}
\newcommand{\sundecided}{\texttt{undecided}\xspace}

\newcommand{\lblamed}{\textsf{LBlamed}\xspace}

\newcommand{\core}{\textsf{core}\xspace}

\newcommand{\lived}{\ensuremath{\Delta_{\text{live}}}\xspace}
\newcommand{\leaderd}{\ensuremath{\Delta_{\text{leader}}}\xspace}
\newcommand{\graced}{\ensuremath{\Delta_{\text{grace}}}\xspace}

\newcommand{\asleep}{\textsf{asleep}\xspace}

\newcommand{\Valid}{\textsf{Valid}\xspace}

\algtext*{EndProcedure}
\algtext*{EndFunction}
\algtext*{EndIf}
\algtext*{EndFor}
\algtext*{EndWhile}

\usepackage{mathtools}

\definecolor{eclipseStrings}{RGB}{42,0.0,255}
\definecolor{eclipseKeywords}{RGB}{127,0,85}
\colorlet{numb}{magenta!60!black}
\lstdefinelanguage{json}{
    basicstyle=\normalfont\ttfamily,
    commentstyle=\color{eclipseStrings},
    stringstyle=\color{eclipseKeywords},
    stepnumber=1,
    numbersep=8pt,
    showstringspaces=false,
    breaklines=true,
    string=[s]{"}{"},
    comment=[l]{:\ "},
    morecomment=[l]{:"},
    literate=
        *{0}{{{\color{numb}0}}}{1}
        {1}{{{\color{numb}1}}}{1}
        {2}{{{\color{numb}2}}}{1}
        {3}{{{\color{numb}3}}}{1}
        {4}{{{\color{numb}4}}}{1}
        {5}{{{\color{numb}5}}}{1}
        {6}{{{\color{numb}6}}}{1}
        {7}{{{\color{numb}7}}}{1}
        {8}{{{\color{numb}8}}}{1}
        {9}{{{\color{numb}9}}}{1}
}

\newcommand\YAMLcolonstyle{\color{red}\mdseries}
\newcommand\YAMLkeystyle{\color{black}\bfseries}
\newcommand\YAMLvaluestyle{\color{blue}\mdseries}
\makeatletter
\newcommand\language@yaml{yaml}
\expandafter\expandafter\expandafter\lstdefinelanguage
\expandafter{\language@yaml}
{
    keywords={true,false,null,y,n},
    keywordstyle=\color{darkgray}\bfseries,
    basicstyle=\YAMLkeystyle,
    sensitive=false,
    comment=[l]{\#},
    morecomment=[s]{/*}{*/},
    commentstyle=\color{purple}\ttfamily,
    stringstyle=\YAMLvaluestyle\ttfamily,
    moredelim=[l][\color{orange}]{\&},
    moredelim=[l][\color{magenta}]{*},
    moredelim=**[il][\YAMLcolonstyle{:}\YAMLvaluestyle]{:},
    morestring=[b]',
    morestring=[b]",
    literate =    {---}{{\ProcessThreeDashes}}3
        {>}{{\textcolor{red}\textgreater}}1
        {|}{{\textcolor{red}\textbar}}1
        {\ -\ }{{\mdseries\ -\ }}3,
}
\lst@AddToHook{EveryLine}{\ifx\lst@language\language@yaml\YAMLkeystyle\fi}
\makeatother
\newcommand\ProcessThreeDashes{\llap{\color{cyan}\mdseries-{-}-}}

\begin{document}

\title{\sysname: Fast and Robust Blockchains through \\ Subsystem Specialization}

\author{Preston Vander Vos$^{2,3}$,
    Alberto Sonnino$^{1,2}$,
    Giorgos Tsimos$^{4}$,
    Philipp Jovanovic$^{1,2}$,
    Lefteris Kokoris-Kogias$^{1}$}

\affiliation{%
    \institution{$^{1}$Mysten Labs;
        $^{2}$University College London;
        $^{3}$Circle;
        $^{4}$Pod Network}
    \country{}
}
\email{}

\begin{abstract}

Blockchain consensus faces a trilemma of security, latency, and decentralization. High-throughput systems often require a reduction in decentralization or robustness against strong adversaries, while highly decentralized and secure systems tend to have lower performance.
We present \sysname, a two-layer consensus architecture to tackle these challenges. 
The main finalization layer, which we call the core layer, is run in a medium-sized validator set. It is instantiated with our novel \lmysticeti consensus protocol. \lmysticeti exposes two client finality paths following the FlexibleBFT~\cite{malkhi2019flexible} paradigm. Fast-path clients that assume $n=5f+1$ experience $2\delta$-finality latency. 
Our experiments show that \lmysticeti reduces latency by $20$--$25\%$ compared to Mysticeti.
Checkpoint-path clients wait for a checkpoint certificate ($4\delta$ latency), which tolerates an additional $\fabc < 2f$ alive-but-corrupt participants---achieving a total of $60\%$ equivocation-fault safety. The key mechanism is a single-proposal-per-height rule that guarantees at most one valid checkpoint exists per slot height, even under heavy equivocation.
The secondary finalization layer, \gossip, is run along a highly decentralized set of guard validators (or full nodes) and provides synchronous decentralized timestamping, proactive misbehavior detection, and a synchronous recovery path. When guard validators observe equivocations or liveness failures they disseminate evidence, agree on misbehaving parties for exclusion or slashing, and either restart the core protocol (for liveness violations) or select a canonical fork (for safety violations). Clients can choose to accept the synchrony assumption as their main assumption and withstand a $n=2f+1$ security level. \gossip can be deployed along any consensus protocol. In \sysname combining \lmysticeti with \gossip  allows clients to choose their security level against their latency requirements. 

\end{abstract}

\maketitle

\pagestyle{\vldbpagestyle}
\begingroup\small\noindent\raggedright\textbf{PVLDB Reference Format:}\\
\vldbauthors. \vldbtitle. PVLDB, \vldbvolume(\vldbissue): \vldbpages, \vldbyear.\\
\href{https://doi.org/\vldbdoi}{doi:\vldbdoi}
\endgroup
\begingroup
\renewcommand\thefootnote{}\footnote{\noindent
    This work is licensed under the Creative Commons BY-NC-ND 4.0 International License. Visit \url{https://creativecommons.org/licenses/by-nc-nd/4.0/} to view a copy of this license. For any use beyond those covered by this license, obtain permission by emailing \href{mailto:info@vldb.org}{info@vldb.org}. Copyright is held by the owner/author(s). Publication rights licensed to the VLDB Endowment. \\
    \raggedright Proceedings of the VLDB Endowment, Vol. \vldbvolume, No. \vldbissue\ %
    ISSN 2150-8097. \\
    \href{https://doi.org/\vldbdoi}{doi:\vldbdoi} \\
}\addtocounter{footnote}{-1}\endgroup
\ifdefempty{\vldbavailabilityurl}{}{
    \vspace{.3cm}
    \begingroup\small\noindent\raggedright\textbf{PVLDB Artifact Availability:}\\
    The source code, data, and/or other artifacts have been made available at \url{\vldbavailabilityurl}.
    \endgroup
}

\section{Introduction}
In currently deployed blockchains it is the protocol, not the client, that decides what finality means.
A transaction is final once it clears a certain point within the security, latency, and decentralization parameter space that the protocol designer fixed in advance.
Yet clients do not value these dimensions equally: a high-value settlement may demand the strongest safety and tolerate waiting, while a latency-sensitive payment prefers the fastest confirmation a weaker assumption allows.
Existing systems make this choice for all clients alike.


FaB~\cite{martin2005fast} showed that lowering the fault-tolerance threshold from $n = 3f + 1$ to $n = 5f + 1$ can reduce finality latency from $3\delta$ to $2\delta$.
Exploring such trade-offs has recently regained traction, with protocols such as Kudzu~\cite{shoup2025kudzu}, Hydrangea~\cite{shrestha2025hydrangea}, and Minimmit~\cite{minimmit}. 
These protocols trade Byzantine resilience for a two-round commit latency, making them more fragile on their fast paths than their $3f+1$ counterparts.
To soften this, Kudzu and Hydrangea retain a secondary commit path alongside the fast one, whereas Minimmit forgoes it entirely.
Crucially, though, the secondary paths of Kudzu and Hydrangea only tolerate additional \emph{crash} faults, i.e., they optimize for more liveness, while assuming the very same Byzantine threshold for safety as the fast path.
Both paths therefore commit under the same safety assumption and clients of these systems are offered no choice over the security--latency trade-off. 
If a client is unwilling to trust the lower resilience assumption, it has no stronger, higher-latency path to fall back on.
This raises an interesting practical question:

\textit{
Can we design a consensus protocol that provides clients with options in the security--latency--decentralization space, finalizing a transaction as soon as the assumptions that the client trusts hold?
}

This paper answers affirmatively by introducing \sysname, a novel dual-layer consensus architecture that synergistically combines the strengths of different consensus paradigms.
Rather than resolving the security--latency--decentralization trilemma at design time on behalf of all clients, \sysname exposes it: each client selects the point it trusts and finalizes at the latency that choice affords.
At \sysname's core is a partially synchronous consensus protocol, \lmysticeti, operated by a medium-sized set of $n_c$ core validators, that trades off failure tolerance ($n_c=5f_c+1$) for lower latency.
It allows clients to decide on $2\delta$ latency with $f_c<\frac{n_c}{5}$, or a $4\delta$ latency with
$f_c<\frac{n_c}{5}$ and an additional $\fabc<\frac{2n_c}{5}$ where $\fabc$ can equivocate but not crash as introduced by FlexibleBFT~\cite{malkhi2019flexible}.

In parallel to \lmysticeti, \sysname runs a secondary synchronous Byzantine Broadcast protocol, which is operated by an enlarged set of validators. We call this enlarged set (that still includes the core validators) the guard validator set. It assumes synchrony and a combined $n=2f+1$ fault mode. Its job is to provide enhanced decentralization, high resilience, and a robust recovery mechanism.
Our approach aims to deliver to the clients a palette of options when it comes to the latency versus security versus decentralization trade-off they wish to follow.

\subsubsection*{The core layer}

Building \sysname's core layer requires a high-performance consensus protocol. Our goal is simple: achieve the lowest possible finality latency and the highest experimentally demonstrated throughput. To do so, we introduce \lmysticeti, which builds on top the partially synchronous DAG-based consensus protocol Mysticeti~\cite{mysticeti} and reduces latency by trading fault tolerance from $n=3f+1$ to $n=5f+1$. \lmysticeti introduces the first $5f+1$ DAG-based consensus protocol and commits transactions in under $0.5$s while sustaining more than $200{,}000$ tx/s (\Cref{sec:evaluation})---a performance level unmatched in the BFT consensus literature.
We also introduce an asynchronous variant of \lmysticeti in \extref{Appendix G}{sec:async}, which tracks the performance of \lmysticeti with a 33\% latency increase.

We use the standard two-step design: first, consensus fixes a single order of operations for everyone, and then identical replicas apply that order ensuring they all end up in the same state (state machine replication). \lmysticeti exposes two explicit finality paths:
\begin{itemize}
    \item {\bf Fast path ($2\delta$):} Clients who follow every block see finality in about $1\,\mathrm{RTT}$ (round-trip time); under $n=5f+1$, this becomes roughly $0.5\,\mathrm{RTT}$ faster than an equivalent $n=3f+1$ protocol. This path is safe as long as at most $f_c$ core validators are Byzantine.
    \item {\bf Checkpoint path ($4\delta$):} Clients willing to wait two additional message delays finalize upon observing a checkpoint certificate (FinalityQC). After a leader slot is committed, validators propose a checkpoint over the resulting state; a critical invariant is that each honest validator proposes \emph{at most one checkpoint per height}. This single-proposal rule guarantees that at most one valid checkpoint can exist per height, even when up to $3f_c$ core validators equivocate. Checkpoint-path clients thus remain safe if fewer than $60\%$ of core validators are equivocating, tolerating an additional $\fabc < 2f_c$ Alive-but-Corrupt~\cite{malkhi2019flexible} faults beyond the $f_c$ Byzantine faults.
\end{itemize}




\subsubsection*{The guard layer}

The guard protocol in \sysname is a highly decentralized auditing and recovery layer, called \gossip. \gossip operates in synchronous steps and 
monitors the core layer for liveness or safety violations (possible with up to $\lfloor 3n_c/5 \rfloor - 1$ faults in the core layer) in the form of equivocations. If this happens, it enables a synchronous recovery protocol.
Upon detecting such misbehavior, the guard validators transition to a recovery phase.
This consists of running Byzantine Broadcast of accumulated evidence and deterministically agreeing on the set of core validators that misbehaved and then excluding them, potentially by slashing their stake.
If the violation relates to liveness, then the reduced set of core validators restarts consensus, whereas if it was a safety violation, they first agree on the canonical fork.
We emphasize that \gossip is not just a simple slashing mechanism, but a full-fledged recovery protocol ensuring that the overall system can continue operating even after a significant attack on the core layer.

\begin{table*}[t]
\centering
\small
\caption{\small Client assumptions and perceived finality latency. $\delta$ denotes the actual message delay, $\Delta$ the synchrony bound used by \gossip. Fast-path clients finalize upon a committed leader slot; checkpoint-path clients finalize upon a FinalityQC (\Cref{sec:checkpoints}). When the core's fast-path threshold is exceeded, \gossip recovers liveness and preserves checkpoint-path safety (\Cref{sec:gossip}).}
\label{tab:comparison}
\begin{tabular}{lcccccc}
\toprule
& \multicolumn{2}{c}{\lmysticeti} & \multicolumn{3}{c}{\sysname (\lmysticeti + \gossip)} \\
\cmidrule(lr){2-3}\cmidrule(lr){4-6}
& Fast path & Checkpoint path & Fast path & Checkpoint path & Synchronous recovery \\
\cmidrule(lr){2-2}
\cmidrule(lr){3-3}
\cmidrule(lr){4-4}
\cmidrule(lr){5-5}
\cmidrule(lr){6-6}
Threshold & $f_c < \frac{n_c}{5}$ & $f_c + \fabc < \frac{3n_c}{5}$ & $f < \frac{n}{6}$ & $f + \fabc < \frac{n}{2}$ & $f < \frac{n}{2}$ \\
Finality & $2\delta$ & $4\delta$ & $2\delta$ & $4\delta$ & $O(\Delta)$ \\
Safety & \checkmark & \checkmark & \checkmark & \checkmark & Recovery preserves Checkpoints \\
\bottomrule
\end{tabular}
\vspace{-0.15cm}
\end{table*}

\subsubsection*{Contributions}
We make the following contributions:
\begin{itemize}
    \item We present \sysname, a novel dual-layer consensus architecture that simultaneously achieves low latency, high security, and high decentralization.
    \item We introduce the two main novel building blocks of \sysname, namely \lmysticeti, a DAG-based consensus protocol that reduces latency by lowering fault-tolerance to $n=5f+1$ and \gossip, a highly decentralized guard protocol for 
          misbehavior monitoring and recovery.
    \item We provide a rigorous security analysis of both \lmysticeti and \gossip, demonstrating that \sysname achieves optimal safety and liveness guarantees under standard assumptions.
    \item We implement and evaluate \lmysticeti and compare it at scale against Mysticeti, confirming the expected latency gains of around $20$--$25\%$ while maintaining similar throughput.
\end{itemize}

\section{System Overview} \label{sec:overview}

Below we provide the system and threat models, the design goals, and a design overview of \sysname.

\subsection{System model} \label{sec:system-model}

We assume a total number of nodes $n$.
\sysname has a \emph{core} and a \emph{guard} layer. All validators are \emph{guard} validators (or guards) but only a subset is \emph{core}.
We denote the number of core validators as $n_c$ and the number of guard validators as $n$.
Core validators continuously operate the core consensus protocol and process transactions, while guard validators checkpoint the output and audit the operation of core validators and help with recovery.
Both core and guard validators are selected using a Sybil-resistant election mechanism~\cite{douceur2002sybil} and we usually assume a proof-of-stake approach~\cite{kucci2021proof}.
We can translate the quantitative assumptions on the nodes to stake assumptions and we often use the two concepts interchangeably.
Core validators are chosen similarly to existing quorum-based blockchains, consisting of roughly the $100$ entities with the highest stake or those that meet specific criteria, such as owning a minimum percentage of the total stake~\cite{sui}.
Guard validators include all stakeholding entities.
The number of guards is assumed to be significantly larger than the number of core validators (multiple hundreds). In practice, we expect full nodes to operate as guards. 
We provide a discussion on the stake distribution in the next section.
\sysname operates as a message-passing system.

\subsection{Threat model} \label{sec:threat-model}

\sysname assumes a computationally bounded adversary, ensuring that common cryptographic security assumptions, like those for hash functions and digital signatures, hold.
We assume a computationally bounded adversary controlling at most $f$ of $n = 2f+1$ total validators (or equivalently, stake $S_f < S/2$ of total stake $S$).
Byzantine nodes may behave arbitrarily; the remaining $n-f$ nodes are honest.
This global bound governs the guard layer; the core layer operates under a stricter threshold ($n_c = 5f_c+1$) that may be temporarily violated, triggering recovery.
The two layers also differ in their network assumptions, detailed below.
We provide a security assessment on the stake requirements for the \lmysticeti and \gossip layers at the end of this section. We also highlight that the two layers of \sysname conceptually operate under different networking models, which we discuss in the following paragraphs.


\para{Core layer} \sysname uses the novel \lmysticeti consensus protocol at its core, which assumes that the total number of core validators $n_c$ satisfies $n_c \geq 5f_c + 1$, where $f_c$ is the maximum number of Byzantine core validators which may deviate from the protocol arbitrarily.
The remaining core validators are assumed to be honest and follow the protocol specification.

\lmysticeti exposes two client finality paths with different security--latency trade-offs, following the FlexibleBFT~\cite{malkhi2019flexible} paradigm.
\emph{Fast-path} clients finalize transactions upon observing a committed leader slot ($2\delta$ latency) and are safe as long as at most $f_c$ core validators are Byzantine.
\emph{Checkpoint-path} clients finalize only upon observing a checkpoint certificate ($4\delta$ latency) and tolerate an additional $\fabc$ \emph{alive-but-corrupt} (AbC) core validators, where $\fabc < 2f_c$.
AbC validators~\cite{malkhi2019flexible} participate in the protocol (so it does not stall) but may equivocate, i.e., sign conflicting messages.
Under the checkpoint path, safety is maintained as long as the total number of equivocating core validators satisfies $f_c + \fabc < 3f_c + 1$, i.e., fewer than $60\%$ of core validators equivocate.
The checkpoint mechanism is formally defined in \Cref{sec:checkpoints} and its safety is proven in \extref{Appendix D.1}{sec:checkpoint-proofs}.
We further assume that \lmysticeti operates in a partially synchronous~\cite{dwork1988consensus} network, where message delays are unbounded until a Global Stabilization Time (GST), after which they are bounded by a constant and known $\Delta$.
 We also provide a variant of \lmysticeti for asynchronous networks in \extref{Appendix G}{sec:async}, where message delays remain unbounded.
Since the threat model of the core layer is significantly weaker than the global assumption, it can fail. However, it will only fail accountably and, as a result, allow the system to detect the adversarial nodes and exclude them, slowly converging the threat model to the operational requirements of the core layer.
The process of monitoring for such accountable faults and initiating recovery is handled by the guard layer.

\para{Guard layer} The guard layer runs \gossip under a synchronous network model with known delay bound $\Delta$.
Crucially, \emph{all validators participate as guards}, including core validators;
the guard set is a superset of the core set, not a disjoint group.
This ensures that corrupting the guard layer requires corrupting the same validators who run the core—an adversary cannot separately target "guard-only" nodes to bypass recovery.
Under synchrony and the global $n = 2f+1$ assumption, \gossip detects equivocations or liveness failures in \lmysticeti, constructs a provable blameset of at least $f_c+1$ misbehaving core validators, and excludes them via Byzantine Broadcast.
While synchrony is a stronger assumption than partial synchrony, \sysname provides optimistic responsiveness: when \lmysticeti's assumptions hold, clients experience fast $2\delta$ or $4\delta$ finality without waiting for synchronous confirmation.

\para{Client Finality Guarantees}
The client finality guarantees are a local assumption. This is on par with designs around Flexible BFT~\cite{malkhi2019flexible, neu2024optimal}. Clients who believe the adversary is not actively compromising \lmysticeti can get the faster finality, whereas the rest can follow the checkpoint finality. Table~\ref{tab:comparison} presents the fault assumptions for clients.

\para{Stake distribution between validator groups} 
We show how to translate the quantitative assumptions from our threat model to stake assumptions and determine the stake distribution between core and guard validators.
We assume there is a total amount of stake $S$ available such that $S=2S_f+1$ where $S_f$ is the stake controlled by the adversary.
We denote the stake assigned to core validators by $S_c$.
Given our initial stake distribution assumptions, the security assumption of the core layer might be violated at some point.
However, in this case, we still want to achieve safety of checkpoints. 
To preserve checkpoint safety, the adversary is allowed to control no more than $60\%$ of $S_c$.
Let $S_{c_f}$ denote the core stake controlled by the adversary; we must maintain the invariant 
\begin{equation}\label{eq:invariant}
S_{c_f}< 3/5 \cdot S_c~.
\end{equation}
Since the adversary can concentrate its full corruption budget $S_f = (S-1)/2$ on the core validators, i.e., $S_{c_f}=(S-1)/2$, and considering that invariant~\ref{eq:invariant} must nevertheless hold, we obtain
\begin{equation}\label{eq:requirement}
    (S-1)/2< 3/5 \cdot S_c \Rightarrow S_c > 5(S-1)/6~.
\end{equation}
In other words, we require that the total stake contributed by the core validators $S_c$ is at least $\frac{5(S-1)}{6}$. 

Consequently, the stake contributed by the guard validators $S_g$ is $S-S_c< \frac{S+5}{6}$.
For simplicity, we can set $S_c = S\cdot5/6 $. 

\subsection{Design goals} \label{sec:goals}

\sysname has the following three primary design goals in terms of performance, security, and decentralization.

\begin{itemize}
    \item[\bf G1] {\bf High performance:} \sysname provides significantly lower latency and equally high throughput as state-of-the-art (DAG-based) consensus protocols when all \lmysticeti underlying assumptions hold. 
    \item[\bf G2] {\bf High security:} \sysname can tolerate up to $S_f$ stake being controlled by the adversary with the total amount of stake being $S = 2S_f + 1$.
    \item[\bf G3] {\bf High decentralization:} \sysname enables non-core-validator participants to contribute to the system's security via \gossip.
\end{itemize}

\subsection{Design overview} \label{sec:design}

We present an overview of \lmysticeti and \gossip.

\subsubsection{The core layer} \sysname introduces the novel consensus protocol \lmysticeti for its core layer. 
\lmysticeti builds atop an uncertified DAG~\cite{mysticeti,cordial-miners,jovanovic2025mahi}, and is the first $5f+1$ DAG-based consensus protocol. It trades off the security threshold from $n = 3f+1$ to $n = 5f+1$ to achieve finality in $2$ message delays instead of the original $3$.
We prove that safety and liveness hold for \lmysticeti in \extref{Appendix D}{sec:proofs}.
The \lmysticeti protocol follows similar ideas to Mysticeti with the following modifications to its decision rules:
\begin{itemize}
    \item {\bf Direct decision rule:} Commit (skip) a leader in round $R$, if it has $4f_c+1$ votes (blames) from round $R+1$.
    \item {\bf Indirect decision rule:} Commit an undecided leader in round $R$, if there is a committed leader in round $R+2$ that supports the round-$R$ leader with at least $2f_c+1$ round-$R+1$ votes.
\end{itemize}

\subsubsection{The guard layer}

\sysname introduces a novel guard protocol, \gossip, which serves three goals:

\para{Fault detection}
This is the common-case execution of \gossip. 
Guard validators synchronously gossip the blocks output by \lmysticeti as a commit sequence and make sure that there is no equivocated block proposal. 
If this is true and no equivocation is detected within the timeout, the \lmysticeti continues uninterrupted. Otherwise, \gossip has to proceed to accountability assignment.


\para{Accountability assignment}
There are two types of accountable faults in \sysname: safety and liveness faults. 
Safety faults are directly detected through the previous mechanism and, once detected, are flagged to all connected peers. 
Liveness faults are only detected when a core validator flags a set of $f_c+1$ core validators as non-responsive and sends this through the gossip network for notarization.
At this point, the guard validators start a timer to see if the DAG building has halted.
If the timer expires, they locally decide that there are liveness faults.

\para{Synchronous recovery} 
When a core validator detects a fault, it starts the synchronous recovery protocol. 
For this, it first decides on an accountable set of faulty validators it wishes to exclude. 
This set is unique and the core validator cannot change it afterwards.
The accountable set comes from either safety faults (i.e., double-spending from the core), or liveness faults as described above. 
In either event, an honest validator (core or guard) can deduce at least an $f_c+1$-sized accountable set.
The validator then Byzantine Broadcasts this set and all validators sign the set if they believe that the set indeed constitutes Byzantine validators. 
The first Byzantine Broadcast that outputs a valid set is considered the last message of the view. 
The validators exclude the blamed stakeholders from the new view and return to regular operation, optionally deciding on the fork they want to adopt.

\section{The \lmysticeti Protocol} \label{sec:lmysticeti}

\lmysticeti is the first $5f+1$ DAG-based consensus protocol. It is built on top of an uncertified DAG, similarly to Mysticeti~\cite{mysticeti} and Cordial Miners~\cite{cordial-miners}. It commits in two message delays by lowering the fault threshold to $f_c < n_c/5$ ($\approx 20\%$). 
We state \lmysticeti's safety and liveness theorems inline with the protocol and defer the full proofs to \extref{Appendix D}{sec:proofs}.
\Cref{alg:main} provides the main entry point and is invoked whenever a validator receives a new block. \Cref{alg:decider} specifies the decision process, with each validator instantiating one \emph{Decider} instance per leader slot. Finally, \Cref{alg:helper} defines utility procedures that support common operations used throughout the protocol. The orange-colored lines of the algorithm refer to the asynchronous version of \lmysticeti which we prove in the  \extref{Appendix G}{sec:async}. This is of independent interest and can be ignored.
Finally, \Cref{alg:checkpoint} defines the checkpointing mechanisms of \lmysticeti.

\begin{algorithm}[t]
    \caption{\lmysticeti}
    \label{alg:main}
    \scriptsize
    \begin{algorithmic}
        \State \texttt{leadersPerRound} \Comment{A number between 1 and $4f_c+1$}
        \State \texttt{waveLength} \Comment{Set to $2$ for \lmysticeti, \altcode{$3$ for \asynclmysticeti}}
        \Statex

        \Procedure{TryDecide}{$r_{committed}, r_{highest}$}
        \State $S \gets [ \; ]$ \Comment{Holds decisions}
        \For{$r \gets r_{highest}$ \textbf{down to} $r_{committed} + 1$}
        \For{$l \gets \texttt{leadersPerRound} - 1$ \textbf{down to} $0$}
        \State $i \gets r \; \bmod$ \texttt{waveLength}
        \State $D \gets$ \Call{Decider}{$i, l$}
        \State $w \gets D.$\Call{WaveNumber}{$r$}
        \State $s \gets D.$\Call{TryDirectDecide}{$w$}
        \If{$s = \bot$} $s \gets D.$\Call{TryIndirectDecide}{$w, S$}
        \EndIf
        \State $S \gets s \parallel S$
        \EndFor
        \EndFor
        \State \Return $S$
        \EndProcedure
        \Statex

        \Procedure{ExtendCommitSeq}{$r_{committed}, r_{highest}$}
        \State $S \gets$ \Call{TryDecide}{$r_{committed}, r_{highest}$}
        \State $S_{commit} \gets [ \; ]$ \Comment{Holds committed blocks}
        \For{$s \in S$}
        \If{$s = \bot$} \textbf{break}
        \EndIf
        \If{$s = \texttt{Commit}(b_{leader})$} $S_{commit} \gets S_{commit} \parallel b_{leader}$
        \EndIf
        \EndFor
        \State \Return \Call{LinearizeSubDags}{$S_{commit}$} \Comment{Same as DAG-Rider~\cite{dag-rider}}
        \EndProcedure

    \end{algorithmic}
\end{algorithm}
\begin{algorithm}[ht]
    \caption{Decider Instance}
    \scriptsize
    \label{alg:decider}
    \begin{algorithmic}

        \State \texttt{waveOffset} $= i$ \Comment{The first parameter of the Decider (i)}
        \State \texttt{leaderOffset} $= l$ \Comment{The second parameter of the Decider (l)}
        \State \texttt{waveLength} \Comment{Set to $2$ for \lmysticeti, \altcode{$3$ for \asynclmysticeti}}
        \Statex

        \Procedure{WaveNumber}{$r$}
        \State \Return $(r - \texttt{waveOffset}) / \texttt{waveLength}$
        \EndProcedure
        \Statex

            \Procedure{ProposeRound}{$w$}
        \State \Return $(w * \texttt{waveLength}) + \texttt{waveOffset}$
        \EndProcedure
        \Statex

        \Procedure{DecisionRound}{$w$}
        \State \Return \Call{ProposeRound}{$w$}$ + (\texttt{waveLength} - 1)$
        \EndProcedure
        \Statex

        \Procedure{StronglyCertifiedLeader}{$w, b_{leader}$}
        \State $B_{decision} \gets$ \Call{GetDecisionBlocks}{$w$}
        \State \Return $|\{ b' \in B_{decision} : $ \Call{IsVote}{$b', b_{leader}$}$ \}| \geq 4f_c + 1$
        \EndProcedure
        \Statex

        \Procedure{SkippedLeader}{$w, b_{leader}$}
        \State $B_{decision} \gets$ \Call{GetDecisionBlocks}{$w$}
        \State \Return $|\{ b' \in B_{decision} : \neg$\Call{IsVote}{$b', b_{leader}$}$ \}| \geq 4f_c + 1$
        \EndProcedure
        \Statex

        \Procedure{TryDirectDecide}{$w$}
        \State $B_{leader} \gets$ \Call{GetLeaderBlocks}{$w$, \texttt{leaderOffset}}
        \For{$b_{leader} \in B_{leader}$}
        \If{\Call{SkippedLeader}{$w, b_{leader}$}} \Return \texttt{Skip}
        \EndIf
        \If{\Call{StronglyCertifiedLeader}{$w, b_{leader}$}} \Return \texttt{Commit}$(b_{leader})$
        \EndIf
        \EndFor
        \State \Return $\bot$
        \EndProcedure
        \Statex

        \Procedure{WeaklyCertifiedLeader}{$b_{anchor}, b_{leader}$}
        \State $w \gets$ \Call{WaveNumber}{$b_{leader}.round$}
        \State $B_{decision} \gets$ \Call{GetDecisionBlocks}{$w$}
        \State \Return $|\{ b \in B_{decision} : $ \Call{IsVote}{$b, b_{leader}$} $\land$ \Call{Link}{$b, b_{anchor}$}$ \}| \geq 2f_c + 1$
        \EndProcedure
        \Statex

        \Procedure{TryIndirectDecide}{$w, S$}
        \State $r_{decision} \gets $\Call{DecisionRound}{$w$}
        \State $s_{anchor} \gets$ first $s \in S$ s.t. $r_{decision} < s.round \land s \neq$ \texttt{Skip}
        \If{$s_{anchor} = \texttt{Commit}(b_{anchor})$}
        \State $B_{leader} \gets$ \Call{GetLeaderBlocks}{$w$, \texttt{leaderOffset}}
        \If{$\exists \; b_{leader} \in B_{leader}$ s.t. \Call{WeaklyCertifiedLeader}{$b_{anchor}, b_{leader}$}} \Return \texttt{Commit}$(b_{leader})$
        \Else \; \Return \texttt{Skip}
        \EndIf
        \EndIf
        \State \Return $\bot$
        \EndProcedure

    \end{algorithmic}
\end{algorithm}
\begin{algorithm}[ht]
    \caption{Helper Functions}
    \scriptsize
    \label{alg:helper}
    \begin{algorithmic}

        \State \texttt{validators} \Comment{The set of validators}
        \State \altcode{\texttt{async} \Comment{Whether the protocol is asynchronous}}
        \Statex

        \Procedure{GetDecisionBlocks}{$w$}
        \State $r_{decision} \gets $\Call{DecisionRound}{$w$}
        \State \Return $DAG[r_{decision}]$
        \EndProcedure
        \Statex

        \Procedure{GetLeaderBlocks}{$w, rank$} \Comment{Validators may equivocate}
        \State $r_{propose} \gets \Call{ProposeRound}{w}$
        \If{\altcode{$\texttt{async}$}}
        \State \altcode{$r_{decision} \gets $\Call{DecisionRound}{$w$}}
        \State \altcode{$s \gets \Call{CombineCoinShares}{\{ b.share \text{ s.t. } b \in DAG[r_{decision}] \}}$}
        \Else
        \State $s \gets r_{propose}$
        \EndIf
        \State $leader \gets \texttt{validators}[(s + rank) \bmod |\texttt{validators}|]$
        \State \Return $\{ b \in DAG[r_{propose}] : b.author = leader\}$
        \EndProcedure
        \Statex

        \Procedure{IsVote}{$b_{support}, b_{leader}$}
        \State /* Note: If \texttt{waveLength} $=2$, equivalent to $\Call{Link}{b_{leader}, b_{support}}$ */
        \Function{VotedBlock}{$b, id, r$} \Comment{Depth-first search}
        \If{$r \geq b.round$} \Return $\perp$ \EndIf
        \For{$b' \in b.parents$}
        \If{$(b'.author, b'.round) = (id, r)$} \Return $b'$ \EndIf
        \State $res \gets \Call{VotedBlock}{b', id, r}$
        \If{$res \neq \perp$} \Return $res$ \EndIf
        \EndFor
        \State \Return $\perp$
        \EndFunction
        \State $(id, r) \gets (b_{leader}.author, b_{leader}.round)$
        \State \Return $\Call{VotedBlock}{b_{support}, id, r} = b_{leader}$
        \EndProcedure
        \Statex

        \Procedure{Link}{$b_{old}, b_{new}$}
        \State \Return $\exists$ a sequence of $k \in \mathbb{N}$ blocks $b_1, \ldots, b_k$ s.t.
        \Statex \hspace{2em} $b_1 = b_{old} \land b_k = b_{new} \land \forall j \in [2, k] : b_j \in \bigcup_{r \geq 1} DAG[r] \land b_{j-1} \in b_j.parents$
        \EndProcedure

    \end{algorithmic}
\end{algorithm}
\begin{algorithm}[ht]
    \caption{Checkpoint Protocol}
    \scriptsize
    \label{alg:checkpoint}
    \begin{algorithmic}
        \State \texttt{committed} $\gets \emptyset$ \Comment{Heights for which a commit was observed}
        \State \texttt{proposed} $\gets \emptyset$ \Comment{Heights for which a ChkProp was broadcast}
        \State \texttt{witnessed} $\gets \emptyset$ \Comment{Heights for which a ChkWitness was broadcast}
        \Statex

        \Procedure{OnCommit}{$B, s$} \Comment{Leader $B$ committed at height $s$}
        \If{$s \notin \texttt{committed}$}
        \State $\texttt{committed} \gets \texttt{committed} \cup \{s\}$
        \State $\sigma \gets \Call{Apply}{B}$ \Comment{Deterministic state transition}
        \State \textbf{sign and broadcast} $\langle \textsc{ChkProp}, s, H(B), \sigma \rangle$
        \EndIf
        \EndProcedure
        \Statex

        \Procedure{OnChkProp}{$\langle \textsc{ChkProp}, s, H(B), \sigma \rangle$}
        \State $\sigma' \gets \Call{Apply}{B}$
        \If{$\sigma' \neq \sigma$ \textbf{or} $H(B)$ malformed} \Return \EndIf
        \State $\textsc{ChkProp} \leftarrow \textsc{ChkProp} \; \cup \; \texttt{proposed}[s]$
        \If{$|\{ p \in \texttt{proposed}[s] : p.\sigma = \sigma \}| \geq 4f_c{+}1$}
        \State $\mathit{CheckpointQC} \gets$ the $4f_c{+}1$ matching proposals
        \State \textbf{sign and broadcast} $\langle \textsc{ChkWitness}, s, H(B), \sigma, H(\mathit{CheckpointQC}) \rangle$
        \EndIf
        \EndProcedure
        \Statex

        \Procedure{OnChkWitness}{$\langle \textsc{ChkWitness}, s, H(B), \sigma, H(\mathit{CheckpointQC}) \rangle$}
        \If{$\textsc{ChkWitness}$ malformed} \Return \EndIf
        \State $\textsc{ChkWitness} \leftarrow \textsc{ChkWitness} \; \cup \; \texttt{witnessed}[s]$
        \If{$|\texttt{witnessed}[s]| \geq 4f_c{+}1$}
        \State $\mathit{FinalityQC} \gets$ the $4f_c{+}1$ matching witnesses
        \State \textbf{return} $\mathit{FinalityQC}$ \Comment{finalise height $s$ with state $\sigma$}
        \EndIf
        \EndProcedure
    \end{algorithmic}
\end{algorithm}

\subsection{The directed acyclic graph} \label{sec:dag}

\lmysticeti builds a \emph{directed acyclic graph (DAG)} of \emph{blocks} that reference each other via cryptographic hashes, similar to Mysticeti.
This DAG enables validators to decide (locally) which blocks to commit and in what order.
The \lmysticeti protocol proceeds in a sequence of logical \emph{rounds}.
In each round, every honest validator proposes exactly one block.

\para{Block creation and validation}
A valid block must include at least the following elements: (1) the author $A$ of the block together with their valid cryptographic signature on the block's contents; (2) a round number $R$; (3) a list of transactions; and (4) at least $4f_c+1$ distinct hashes of valid blocks from the previous round $R-1$.
%

Honest validators store only valid blocks in their local DAG and discard invalid ones. Moreover, they reference a block's hash in their proposals only if the block is valid and after downloading and verifying its entire causal history, thereby ensuring the correctness of the block's lineage. In particular, honest validators always include a reference to their own block from the previous round.

\para{Rounds and waves}
As already mentioned, \lmysticeti operates in \emph{rounds} and two subsequent rounds $R$ and $R+1$ form a \emph{wave}. \Cref{fig:rounds} (left) illustrates an example of a single wave with six validators $(v_0, v_1, v_2, v_3, v_4, v_5)$.
The first round $R$ (\propose) of a wave contains the blocks that the wave attempts to commit ($P_0, P_1, P_2, P_3, P_4, P_5$) along with the equivocating block $P_1'$.
In the second round $R+1$ (\decision), each block serves as a \emph{vote} for the \propose blocks it references.
In the example of \Cref{fig:rounds}, blocks $V_0, V_1, V_2, V_3,$ and $V_4$ vote for $P_0, P_1, P_2, P_3,$ and $P_4$ (but not for $P_1'$ or $P_5$), whereas block $V_5$ votes for $P_1', P_2, P_3, P_4,$ and $P_5$ (but not for $P_0$ or $P_1$). The procedure $\Call{IsVote}{\cdot}$ in \Cref{alg:helper} 
formally defines a vote.
A \propose block is considered \emph{strongly certified} if it has at least $4f_c+1$ votes.
We say that the voting blocks form a \emph{strong certificate} for the proposed block.
In the example, blocks $P_0, P_1, P_2, P_3,$ and $P_4$ are strongly certified, while $P_1'$ and $P_5$ are not.
If a block receives at least $2f_c+1$ votes but fewer than $4f_c+1$, it is considered \emph{weakly certified}.
As before, we also say that the voting blocks form a \emph{weak certificate} for the proposed block.
As shown in \Cref{fig:rounds} (right), \lmysticeti  initiates a new wave every round: round $R$ is a \propose round for wave $1$, round $R+1$ is a \decision round for wave $1$ and a \propose round for wave $2$, etc.
\Cref{alg:decider} 
formally defines a wave.

\begin{figure*}[t]
    \centering
    \includegraphics[width=0.70\textwidth]{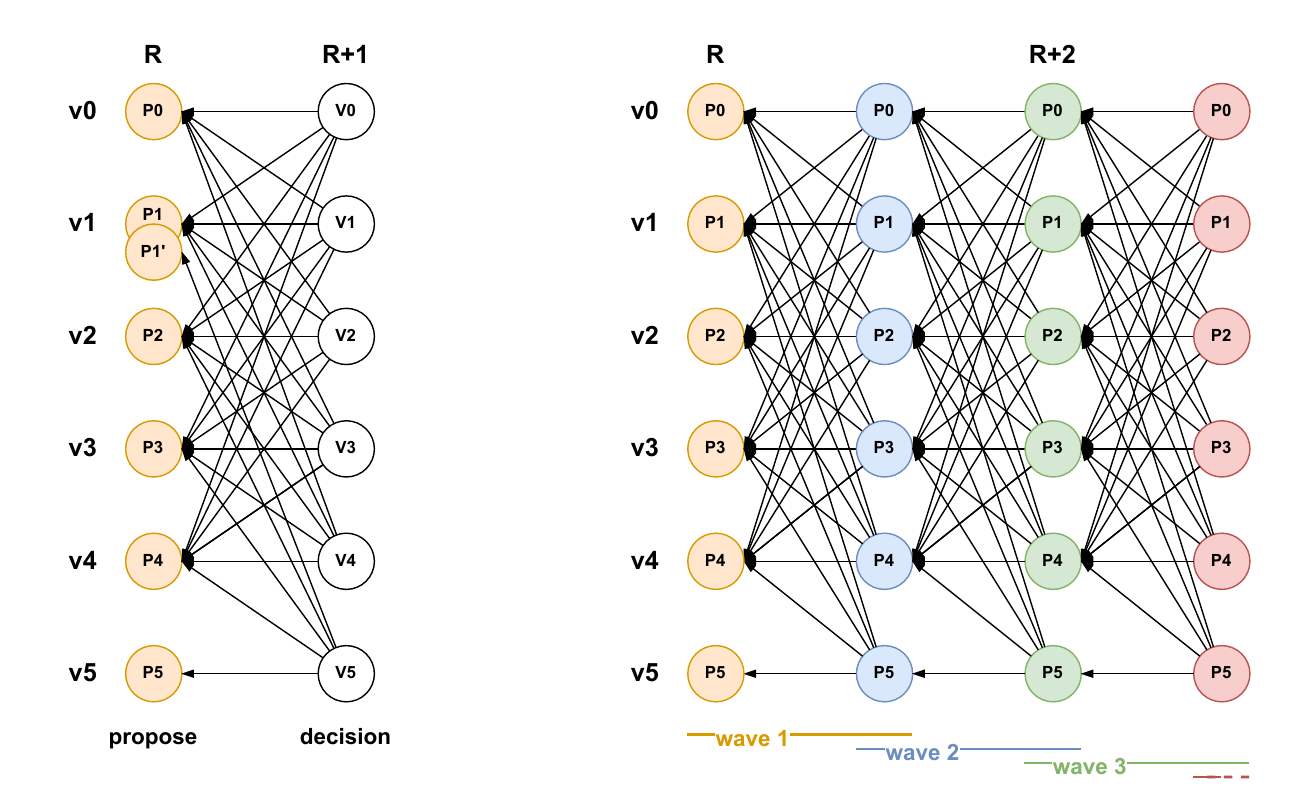}
    \caption{ \footnotesize
        The structure of the \lmysticeti DAG. Left: A wave consisting of two rounds (\propose and \decision). Right: Wave patterns in the \lmysticeti protocol (each round initiates a new overlapping wave).
    }
    \label{fig:rounds}
\end{figure*}

\subsection{Proposers and anchors}

\lmysticeti's leader slots represents a pair (validator, round) and can be either empty or contain the validator's proposal for the respective round. If the validator is Byzantine, they may have equivocated, in which case the slot would contain more than one block.

Multiple leader slots can be instantiated per round, enabling parallel leader block proposals. Each slot is in one of three states: \scommit, \sskip, or \sundecided. All slots begin in the \sundecided state, and the protocol's objective is to classify them as either \scommit or \sskip. The commit state indicates that the slot's block should be included in the total ordering, while the \sskip state allows the protocol to exclude slots from crashed or Byzantine validators. Crucially, the \sundecided state forces subsequent leader slots to wait, preventing unsafe commitments. Similarly to related work \cite{mysticeti,shoal++,jovanovic2025mahi,sailfish}, the number of leader slots per round is a system parameter. Before advancing round, validators wait up to $2 \Delta$ time to receive the blocks from the leaders slots of the previous round, if they have not already.

\subsection{The decision rule} \label{sec:decision-rule}

We present the decision rule of \lmysticeti through an example protocol run.
\Cref{fig:example} shows the local view of a \lmysticeti validator in a system with six validators $(v_0, v_1, v_2, v_3, v_4, v_5)$, parameterized with two leader slots per round. We denote a block as $B_{(v_i, R)}$, where $v_i$ is the issuing validator and $R$ is the block's round.
Initially, all proposer slots are in the \sundecided state. The validator examines the portion of the DAG shown in \Cref{fig:example} and attempts to classify as many leader-slot blocks as possible into either \scommit or \sskip.

\begin{figure}[t]
    \centering
    \includegraphics[width=0.70\columnwidth]{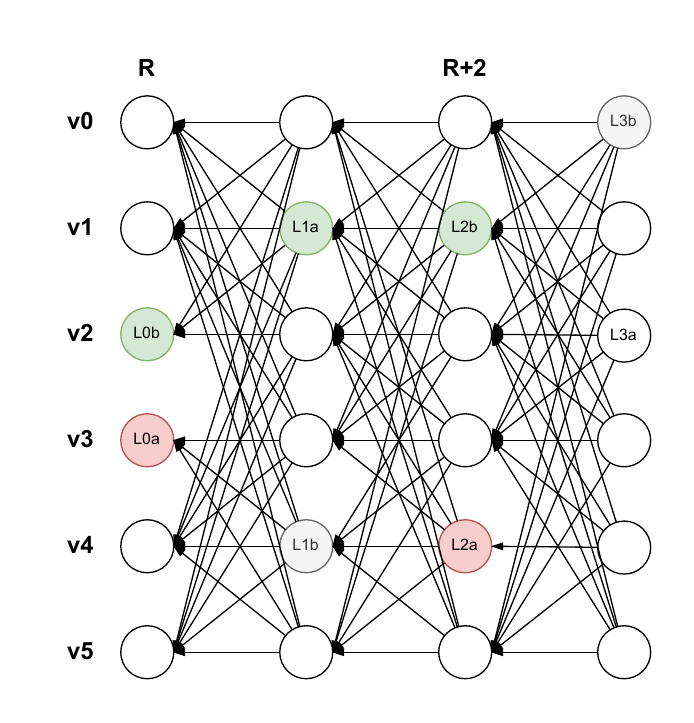}
    \caption{\footnotesize Example with six validators and two leader slots per round.}
    \label{fig:example}
\end{figure}

\para{Step 1: Determine the leader slots}
The sequence of leader slots for each round is predefined and known to all validators, following a deterministic round-robin rotation. In the example of \Cref{fig:example}, every validator knows in advance that the leader slot sequence is [$L_{0a}$, $L_{0b}$, $L_{1a}$, $L_{1b}$, $L_{2a}$, $L_{2b}$, $L_{3a}$, $L_{3b}$].
This deterministic mechanism ensures that even if validators have divergent DAG views, they still agree on the leader slots (and their order) for a given round—regardless of whether a block has been observed in that slot. This design also enables low latency by allowing multiple leaders per round, using these slots to order the causal history.

\para{Step 2: Direct decision rule}
The validator attempts to classify each slot (even those without a block) as either \scommit or \sskip. To do so, it processes each slot individually, starting with the lowest slot that is part of a complete wave ($L_{2b}$), applying the \lmysticeti \emph{direct decision rule}. The validator classifies a block $B$ in a slot as \sskip if it observes $4f_c+1$ blocks from the subsequent \decision round that do not vote for it, and as \scommit if it is strongly certified. As discussed in \Cref{sec:dag}, a block $B$ is strongly certified if there are $4f_c+1$ blocks voting for it. Otherwise, the validator leaves the slot as \sundecided.

In this example, the validator first analyzes $L_{2b}$ and observes that $B_{(v_0, R+3)}$, $B_{(v_1, R+3)}$, $B_{(v_2, R+3)}$, $B_{(v_3, R+3)}$, and $B_{(v_5, R+3)}$ strongly certify it.
Therefore, it classifies $L_{2b}$ as \scommit.
\Cref{sec:evaluation} shows that this scenario is the most common (in the absence of an adversary) and results in the lowest latency.
The validator then analyzes $L_{2a}$ and observes that $B_{(v_0, R+3)}$, $B_{(v_1, R+3)}$, $B_{(v_2, R+3)}$, $B_{(v_3, R+3)}$, and $B_{(v_5, R+3)}$ do not vote for it.
Therefore, it classifies $L_{2b}$ as \sskip.
The presence of $4f_c+1$ blocks from the \decision round that do not vote for a block ensures that it will never be certified, and will thus never be committed by other validators with a potentially different local view of the DAG.
\Cref{sec:evaluation} shows that this rule allows \lmysticeti to promptly skip (benign) crashed leaders to minimize their impact on the protocol's performance.

\para{Step 3: Indirect decision rule}
In the case where the direct decision rule cannot classify a slot, the validator uses the \lmysticeti \emph{indirect decision rule}. This rule looks at future slots to decide about the current one. First, it finds an \emph{anchor}. This is the earliest leader slot with a round number $R' > R + 1$ that is either still classified as \sundecided or already classified as \scommit. If the anchor is \sundecided, the validator marks the current slot as \sundecided. If the anchor is \scommit, the validator checks if it references at least one weak certificate over the current slot. If it does, the validator marks the current slot as \scommit. If it does not, the validator marks the current slot as \sskip. We show in \extref{Appendix D}{sec:proofs} that the direct and indirect decision rules are consistent, namely if one validator direct commits a block no honest validators will indirect skip it and vice versa.

In this example, the validator fails to classify $L_{0b}$ using the direct decision rule as it is neither strongly certified nor are there enough blocks from the subsequent \decision round that do not reference it to classify it as \sskip. It thus searches for its anchor. Since $L_{2a}$ has been classified as \sskip, it cannot serve as an anchor; therefore, $L_{2b}$ is the anchor for $L_{0b}$. Given that $L_{2b}$ references $B_{(v_0, R+1)}$, $B_{(v_1, R+1)}$, and $B_{(v_2, R+1)}$, which form a weak certificate for $L_{0b}$, the validator classifies $L_{0b}$ as \scommit.
Finally, $L_{0a}$ cannot be classified using the direct decision rule as well. Its anchor is again $L_{2b}$, but it does not reference a weak certificate over $L_{0a}$. Therefore, the validator classifies $L_{0a}$ as \sskip.

\para{Step 4: Commit sequence}
After processing all slots, the validator derives an ordered sequence of the leader-slot blocks. It then iterates over this sequence, committing all slots marked as \scommit and skipping all slots marked as \sskip, until it encounters the first \sundecided slot. As shown in \extref{Appendix D}{sec:proofs}, this commit sequence is safe, and eventually, every slot is classified as either \scommit or \sskip.
In the example of \Cref{fig:example}, the leader sequence is [$L_{0b}$, $L_{1a}$].

Following the approach introduced by DagRider~\cite{dag-rider}, the validator derives the final commit sequence by linearizing the blocks within the sub-DAG defined by each leader block using a depth-first search. If a block has already been included by a previous leader slot, it is not re-linearized. Leader slots are processed sequentially, ensuring that all blocks appear in the final commit sequence exactly once and in an order consistent with their causal dependencies. The procedure $\Call{LinearizeSubDags}{\cdot}$ (\Cref{alg:main}) formalizes this step.
In the running example, the commit sequence is [$L_{0b}$, $B_{(v_0, R)}$, $B_{(v_1, R)}$, $B_{(v_4, R)}$, $B_{(v_5, R)}$, $L_{1a}$].

\subsection{Client finality paths and checkpoints} \label{sec:checkpoints}

\lmysticeti exposes two finality paths to clients, each offering a different trade-off between latency and fault tolerance.
Both paths use the same quorum threshold of $4f_c+1$. \Cref{alg:checkpoint} formalizes the checkpoint protocol.

\para{Fast path ($2\delta$ finality)}
A fast-path client finalizes a transaction as soon as the leader slot containing it is committed via the direct or indirect decision rule described above.
This requires a single wave (two message delays) in the common case when the leader is directly committed.
Fast-path safety holds as long as at most $f_c$ of the $n_c = 5f_c+1$ core validators are Byzantine\iflong\ (\Cref{thm:safety}); liveness after GST is established by \Cref{thm:liveness} and agreement by \Cref{thm:agreement} (see \Cref{sec:proofs})\else; fast-path safety, liveness after GST, and agreement are established in \longversion\fi.

\para{Checkpoint path ($4\delta$ finality)}
Checkpoint-path clients trade latency for stronger fault tolerance by waiting for a \emph{checkpoint certificate}.
The checkpoint mechanism chains two additional certificate steps on top of the commit path, as follows.

\para{Step 5: Checkpoint proposal}
After a leader slot at height $s$ is committed (i.e., a strong certificate exists for the leader block~$B$ at round~$R$), each validator computes the deterministic state transition $\sigma = \textsc{Apply}(B)$ and then broadcasts a \emph{checkpoint proposal} $\langle \textsc{ChkProp}, s, H(B), \sigma \rangle$.
Crucially, \textbf{an honest validator broadcasts at most one checkpoint proposal per slot height}, even if it later observes a competing committed branch. This single-proposal rule is the key invariant that prevents conflicting checkpoints.
When $4f_c+1$ matching checkpoint proposals for the same $(s, H(B), \sigma)$ are collected, they form a \emph{CheckpointQC}.

\para{Step 6: Checkpoint witness}
If a validator observes a valid CheckpointQC for height $s$, it broadcasts a \emph{checkpoint witness}\\ $\langle \textsc{ChkWitness}, s, H(B), \sigma, H(\text{CheckpointQC}) \rangle$.
An honest validator broadcasts a witness as long as the checkpoint's state $\sigma$ matches its own deterministic execution of the payload, even if its single permitted checkpoint proposal for that height was already spent on a different, unfinalized branch.
When $4f_c+1$ matching witnesses are collected, they form a \emph{FinalityQC}.
A checkpoint-path client finalizes height $s$ upon observing the FinalityQC together with its underlying CheckpointQC.

\para{Latency derivation}
The four message delays of this path are:
\begin{enumerate}
    \item The leader proposes a block (round $R$).
    \item Validators vote; $4f_c+1$ votes yield a strong certificate, committing the leader (round $R+1$). \emph{Fast-path clients finalize here.}
    \item Validators broadcast checkpoint proposals; $4f_c+1$ matching proposals form a CheckpointQC (round $R+2$).
    \item Validators broadcast checkpoint witnesses; $4f_c+1$ matching witnesses form a FinalityQC (round $R+3$). \emph{Checkpoint-path clients finalize here.}
\end{enumerate}
In practice, checkpoint proposals and witnesses are embedded as metadata in subsequent DAG blocks (rounds $R+2$ and $R+3$) rather than sent as separate messages, keeping the mechanism integrated into the DAG structure.

\para{Checkpoint-path fault tolerance}
The checkpoint path tolerates an additional $\fabc < 2f_c$ alive-but-corrupt (AbC) validators~\cite{malkhi2019flexible}---validators that participate in the protocol but may equivocate.
Under this extended threat model, the total number of equivocating validators can reach up to $f_c + \fabc < 3f_c + 1$, i.e., fewer than $60\%$ of the $5f_c+1$ core validators.
Safety is guaranteed because the single-proposal rule ensures that at most one CheckpointQC can exist per height\iflong\ (\Cref{lem:unique-checkpoint}, \Cref{sec:proofs})\else\ (shown in \longversion)\fi, even when up to $3f_c$ validators equivocate\iflong; the end-to-end checkpoint-path safety statement is \Cref{thm:resilient-safety}\fi.
If the fast-path safety threshold is exceeded (i.e., more than $f_c$ validators equivocate) and a fork occurs at the commit level, the checkpoint path remains safe: honest replicas halt and rely on \gossip's recovery mechanism (\Cref{sec:gossip}) to resolve the fork, using the unique CheckpointQC chain as the canonical history.

\section{The \gossip protocol} \label{sec:gossip}

While the previous section focuses on how \lmysticeti achieves low latency, this holds under the assumption of fewer corruptions and medium-level of decentralization, 
which introduces the practical risk of a weaker resilience threshold. 
We now present an approach to mitigate any security concerns, by introducing an additional (slow) layer of higher resilience and decentralization.
That is \gossip, a protocol that monitors the operation of \lmysticeti in a synchronous pace, and corrects potential violations via two guarantees; 
(i) any safety or liveness violation will be caught and 
(ii) at least $f_c+1$ violating core validators will be provably identified by all participating validators.
In practice, \gossip will operate in a synchronous, slower pace than \lmysticeti, while verifying that the main protocol is operating correctly, 
and recovering liveness and safety in case any of the two fails due to more than $f_c$ corruptions. All \lmysticeti validators are also \gossip validators.
We present the main \gossip functionality in Algorithm~\ref{alg:guard} and some helper functions in Algorithm~\ref{alg:guard-helper}.
While \gossip is specifically designed for \lmysticeti, its framework can be generalized to be deployed along different consensus protocols.
The existence of \gossip allows \sysname to have a checkpoint path ($4\delta$) resilient for up to $f=\frac{n}{2}$ malicious validators since its liveness is also guranteed through recovery (see~\Cref{tab:comparison}).

\subsection{Motivation}
The \gossip protocol ensures safety and liveness of \sysname when the system is under attack. Clients can opt to utilize \gossip and exchange latency to gain in resilience in scenarios  with higher decentralization.
We argue via two constructive examples. 

\emph{Example 1: Liveness failure.} 
Assume that during the operation of \lmysticeti, the resilience assumption breaks, and an adversary \A takes control of $f_c+1$ validators.
In such a case, it is simple to see that our protocol loses liveness indefinitely, even via crash-faults only.
The protocol requires $4f_c+1$ votes for each round's slots to be determined and for it to move forward to the next round.
With fewer than $4f_c+1$ honest validators participating, the protocol could be skipping rounds or make no progress indefinitely.

\emph{Example 2: Safety failure.} 
Again, assume that during the operation of \lmysticeti, the resilience assumption breaks, and an adversary \A takes control of $3f_c$ core validators.
The adversary is then able to present to some validators a directly committed leader block $B_L$ ($4f_c+1$ support), whereas for others it will be undecided ($f_c+1$ support). The next leader will also be presented the latter case, so in its view $B_{L+1}$ will not have enough support to indirectly commit $B_L$.
When $B_{L+1}$ is directly committed, the validators will diverge views (some will already have $B_L$ committed whereas others will skip it) and safety will break for clients on the fast-path\footnote{Notice that checkpoint path is still safe and anyone finalizing there will be secure.}. 


Both examples showcase the importance of maintaining the resilience assumption for \lmysticeti to operate correctly.
However, in practice, it is possible that an adversary could temporarily break the resilience assumption.
In such a case, we would like to ensure that our protocol can recover both safety and liveness,
which is our main goal with \gossip.
\change{\gossip addresses both failure modes by providing \emph{accountable safety} and \emph{accountable liveness} (Definitions~\ref{def:accountable-safety}--\ref{def:accountable-liveness}): within bounded time after any violation, all honest guards agree on a provable blameset of $\ge f_c+1$ faulty validators, enabling the recovery procedure of \Cref{thm:recovery}.}


\subsection{\gossip setting}

Before diving into the protocol details, we recall the overall setting:
We assume a set of $n$ validators who in total hold stake $S$ such that 
each validator maintains stake equivalent to $S/n$
\footnote{This model can be easily simulated even in a setting where validators hold unequal units of stake, by having each validator simulate a separate Sybil validator per each unit of stake it holds.}.
We assume that there exists a polynomially bounded adversary \A who can corrupt up to a total of $S_f\le (S-1)/2$ stake in a static fashion, 
i.e., \A picks which validators to corrupt (up to the corruption threshold) before the protocol starts, and cannot change corruptions after.

For this part of our construction, we assume that validators communicate over a synchronous, point-to-point protocol, with a known network delay $\Delta_{net}$. We also denote by $\Delta$, the \emph{network delay for reliable broadcast}. This is the theoretical delay with which we will analyze our guard protocol\footnote{In practice, we will set $\Delta$ to a sufficiently large time so that all messages are guaranteed to be delivered by then.
Since this construction is going to be part of the slow path, we can set $\Delta$ to a large number.}.
\change{We write $\Call{BC}{i, m, \sigma}$ for the synchronous \emph{Byzantine Broadcast} primitive used by guard $i$ to disseminate message $m$ with signature $\sigma$; under this synchrony assumption, all honest parties receive $m$ within $\Delta$ (agreement, validity, termination).}
This is also the $\Delta$ timeout used for the liveness proofs of \lmysticeti, however, in the practical deployment we optimistically wait a much smaller timeout before blaming a leader. This practice has been introduced by Shoal~\cite{shoal}, where the real $\Delta$ timeout is only used when many consecutive leaders are skipped. It balances well the practical needs with the theoretical requirements and is what most production systems follow.


\subsection{Monitoring failures} 

\para{Block handling} 
Upon receiving a block, each guard verifies that the block is valid with respect to \lmysticeti, 
by executing the predetermined validity checks defined by the \lmysticeti protocol.
Furthermore, guards check for equivocations with respect to every received block. 
They do so, by maintaining state of what blocks they have received from which core validators for each round of \lmysticeti.
If a block is equivocated by its authoring validator, this is proof of misbehavior for that party.
In the current construction, guards simply do not consider any such blocks at all.
However, more sophisticated designs could allow such blocks to contribute to the provable equivocations of the system, so that malicious parties are identified via such equivocations as well.

\para{Blamesets}
Blamesets are the main gadget of our construction.
A \emph{valid blameset} is a set of at least $f_c+1$ core validators, that have provably misbehaved. 
There are two types of blamesets, namely \emph{liveness} and \emph{safety} blamesets, according to the type of misbehavior the validator has been proven to have performed. 
Each valid blameset is accompanied by a valid proof of misbehavior for each validator in it.
Once an honest guard has established a new valid blameset, it can request a recovery process, via which, guards will agree on one blameset, and slash the validators in that blameset.

\para{Liveness failure detection}
Guards are also tasked with restoring the core protocol's liveness.
To check for potential liveness failures, for each new round, they maintain a timer $\lived$ and a set $\asleep$ that contains inactive core validators and is initialized to $\core$, the set of all core validators.
Upon receiving a valid block from some core validator $v_i$ for round $r$, the guard will remove $v_i$ from $\asleep(r)$, indicating that the validator is active for the round.
After the timer expires, any validator remaining in $\asleep$, will be blamed by the guard. 
After the blaming phase, validators have a chance to respond within a grace period $\graced$.
During the grace period, \change{any validator can exonerate a blamed party by relaying} a valid block for the round produced by the blamed validator, to convince all others that the blamed party is live.
After the grace period, if a majority of guard stake has attested blame for a validator (satisfying the $\Call{Maj}{\cdot}$ condition), that validator is included in the liveness blameset and recovery is triggered.

\para{Safety failure detection}
Honest validators are guaranteed to not equivocate, and to not vote for equivocating blocks. 
However, if the safety assumption of the \lmysticeti protocol breaks, i.e., the adversary can corrupt more than $f_c$ validators, then equivocating blocks can be committed (still, no two equivocating blocks can both be directly committed, as long as the corruptions remain $\le 3f_c$). However, as we described before there could be an equivocation through the indirect skip path. 
In this case, \gossip will come in play; every time a guard observes a new set of blocks that are being committed to the DAG, they compare existing already committed, leader blocks to the new ones.
In case a guard finds two equivocating blocks, it compares their support, to find the set of overlapping parties supporting the two. 
This overlapping set, which as we show will be provable and will contain at least $f_c+1$ misbehaving parties, can then be used for recovery/slashing.
Equivocation is meant in a broad sense and while we focus on \lmysticeti this generalizes to any consensus protocol where equivocating blocks or votes can exist.


\subsection{Recovery}
In this construction, we show how guards can catch liveness/safety issues, and in such cases, agree on a blameset of at least $f_c+1$ Byzantine parties.
Honest guards can agree simply by running Byzantine agreement on the recovery set of each guard, which they can do since they maintain honest majority.
After that, they can deterministically choose the first (in order) valid blameset, as the set of core validators to be removed.

After honest guards agree on such a blameset, they can provably notify the core validators.
They in turn can then disregard the participants proven to be malicious (or in practice, slash their stake, which leads to future research towards cryptoeconomic/incentive-based security of such systems).
The core validators can, at that point, execute an honest-majority agreement protocol, since they will have regained majority (by removing $f_c+1$ malicious out of $5f_c+1$ total, the new split is $2f_c+1$ honest out of $4f_c$ total).
For this reason, we focus on showing how to regain honest majority in the core validators in cases of misbehaviour, and not on the specifics of how to utilize the honest majority to recover the protocol afterwards.

\change{\para{Recovery protocol, fork-choice rule, and client guarantees}
Once all honest guards agree on blameset $\mathcal{B}$ (Theorems~\ref{thm:guard-accountable-safety}--\ref{thm:guard-accountable-liveness}), recovery proceeds in three steps: \emph{(R1)} the canonical chain is the unique CheckpointQC of highest height with a FinalityQC (see~\extref{Theorem~8}{thm:resilient-safety}); \emph{(R2)} removing $\mathcal{B}$ yields $\mathcal{V}'=\mathcal{V}\setminus\mathcal{B}$ with at most $4f_c$ validators and a strict honest majority ($\geq 2f_c+1$, enabling a synchronous honest-majority protocol~\cite{abraham2020synchotstuff}); \emph{(R3)} all state finalized by a checkpoint-path FinalityQC is preserved.
More details and comparison with prior recovery work~\cite{civit2025recover,gong2025recovery,sheng2023bft} appear in \extref{Appendix E}{sec:guard-proofs}.}




\begin{algorithm}[t]
\caption{\gossip Main Functions (code for guard $i$)}\label{alg:guard}
\begin{algorithmic}
\scriptsize
\Statex \change{\textbf{Local state:}}
\State \change{$\asleep(r)$: core validators not yet seen active in round $r$; initialized to \core}
\State \change{$blames(k,r)$: guards that blamed $\core_k$ for round $r$; initialized to $\emptyset$}
\State \change{$\lblamed(r)$: liveness-blamed core validators for round $r$, with proofs $\pi_r$; initialized to $\emptyset$}
\State \change{$S_i$: blameset $i$ has committed to in recovery; initialized to $\bot$}
\State \change{$now$: highest round whose liveness $i$ is currently monitoring}
\Statex

\Procedure{OnRound}{$r$}
    \State $now \gets r-1$; set timers $\lived =4\Delta$, $\leaderd=2\Delta$
    \If{$\leaderd$ fires}
        \State $now \gets r$
        \For{$l\in \Call{GetLeaders}{r-1}:l\in\asleep(r-1)$}
            \State $\Call{BC}{\mathsf{LBlame},i, l,r-1}$
        \EndFor
    \EndIf
    \If{$\lived$ fires}
        \For{$\core_j\in\asleep(r)$}
            \State \Call{BC}{$\mathsf{LBlame},i, \core_j, r$}
        \EndFor
        \For{$l\in \Call{GetLeaders}{r-1}:l\notin\asleep(r-1)$}
            \For{$\core_j: \neg \Call{IsVote}{b_{(j,r)}, b_{(l,r-1)}}$}
                \State $\Call{BC}{\mathsf{LBlame},i, \core_j,r}$
            \EndFor
        \EndFor
        \State \change{set timer $\graced =2\Delta$} \Comment{\change{grace period; $\graced$ expires $6\Delta$ after round start}}
    \EndIf
    \If{$\graced$ fires and $\lvert\lblamed(r)\rvert\ge f_c+1$}
        \State\change{\Call{Recover}{$\lblamed(r), \pi_{r}$}}
    \EndIf
\EndProcedure
\Statex

\Procedure{OnBlock}{$b:=B_{(u,r)}$}
      \If{$\Valid(b) \land \change{\Call{Equivocates}{b, \texttt{LocState()}}=\bot} \land now \le r$}
        \State $\asleep(r)\gets\asleep(r)\setminus \{u\}$
        \State \Call{BC}{$i, b,  \sigma_i(b)$}
      \EndIf
\EndProcedure
\Statex

\Procedure{OnCoreUpdate}{$S$} \Comment{$\bot \neq S\gets\Call{TryDecide}{\dots}$}
    \State $(Set,\pi)\gets\Call{CheckEquivocation}{S}$
    \If{$Set\neq\bot$}
        \State \change{\Call{Recover}{$Set, \pi$}}
    \Else
        \For{$b \in S$}
            \State \Call{BC}{$i, b,  \sigma_i(b)$}
        \EndFor
    \EndIf
\EndProcedure
\Statex

\Procedure{OnLBlame}{$j, \core_k,r$}
    \State $blames(k,r)\gets blames(k,r)\cup\{j\}$
    \If{$\Call{Maj}{blames(k,r)}$}
        \State $ \lblamed(r)\gets  \lblamed(r)\cup\{k\}$; $\pi_r\gets\pi_r\Vert blames(k,r)$
    \EndIf
\EndProcedure
\Statex

\Procedure{OnRecover}{$j,S_B,\pi$}
    \If{\change{$\mathsf{isValidBlameSet}(S_B, \pi)$}}
        \State \change{\Call{Recover}{$S_B, \pi$}}
    \EndIf
\EndProcedure

\end{algorithmic}
\end{algorithm}

\begin{algorithm}[t]
    \caption{\gossip Helper Functions}
    \scriptsize
    \label{alg:guard-helper}
    \begin{algorithmic}

        \State $\texttt{validators}$ \Comment{The set of validators}
        \State \change{$\texttt{leadersPerRound}$} \Comment{\change{Number of leader slots per round}}
        \Statex

\Procedure{GetLeaders}{$r$}
    \State Let $l = \texttt{leadersPerRound}$, $V = |\texttt{validators}|$
    \State \Return \{$\texttt{validators}[(r + d) \bmod V]: d = 0,\dots, l-1\}$
\EndProcedure
\Statex

\Procedure{\change{IsVote}}{$\change{b,\, b'}$}
    \State \change{\Return\ $b' \in \mathsf{parents}(b)$} \Comment{\change{$b$ votes for $b'$ iff $b'$ is a direct causal parent of $b$ in the DAG}}
\EndProcedure
\Statex

\Procedure{ResolveEquivocation}{$b,b'$}
  \State  $Set\gets \{v\in \texttt{validators}: \text{ supporting }b\text{ and }b'\}$
  \State $\pi\gets\{\text{respective blocks showing support from }Set\}$
  \State\Return $(Set,\pi)$
\EndProcedure
\Statex

\Procedure{CheckEquivocation}{$S$}
  \If{$\exists b\in S: B\gets\Call{Equivocates}{b, \texttt{LocState()}}\land B\neq\bot$}
      \State $(Set,\pi)\gets \Call{ResolveEquivocation}{B,b}$
      \State \Return $(Set,\pi)$
  \EndIf
  \State \Return $\bot$
\EndProcedure
\Statex

\Procedure{Recover}{\change{$BS, \pi$}}
    \If{\change{$S_i \neq \bot$}} \Return \Comment{\change{committed to a blameset already; ignore duplicate triggers}}
    \EndIf
    \State \change{$S_i \gets BS$}
    \State \Call{BC}{$i, BS, \sigma_i(BS, \textsf{recover})$}
    \State Set $\mathsf{RCVec}_i = \langle S_k\rangle_{k\in[n]}$, \change{where}
    \[ \change{S_k = \begin{cases}
        S_k, & \text{if a unique }(k, S_k, \sigma_k(S_k,\textsf{recover}))\text{ arrived within }\Delta,\\
        \bot, & \text{otherwise.}
       \end{cases}}
    \]
    \State $\Call{BA}{k, S_k} \to SR_k$, for all $k\in[n]$
    \State $\Call{Order}{\{SR_k\}_{k\in [n]}} \to \mathbf{O}$
    \State \Return the first valid $ SR_k\in \mathbf{O}$
\EndProcedure
\end{algorithmic}
\end{algorithm}



\change{%
\subsection{Accountability Definitions}

We now state the formal properties that \gossip is designed to provide.
These properties extend the classical safety and liveness of \lmysticeti with \emph{accountability}: when a violation occurs, it is detected and attributed to a set of provably misbehaving validators.

\begin{definition}[Accountable Safety]\label{def:accountable-safety}
    A guard protocol $\Pi_g$ provides \emph{accountable safety} for \lmysticeti if the following holds:
    if two honest core validators ever disagree on the commit/skip decision for any leader slot
    (i.e., a safety violation occurs in \lmysticeti), then within $2\Delta + \Delta_{BA}$ of the moment the first honest party observes the violation, all honest \gossip validators agree on a \emph{safety blameset} $\mathcal{B}$ with $|\mathcal{B}| \geq f_c+1$, where every $v \in \mathcal{B}$ is accompanied by cryptographic proof of equivocation \change{(and hence $\mathcal{B}$ contains no honest validator, by \extref{Lemma~16}{lem:guard-honest-safe})}.
\end{definition}

\begin{definition}[Accountable Liveness]\label{def:accountable-liveness}
    A guard protocol $\Pi_g$ provides \emph{accountable liveness} for \lmysticeti if the following holds:
    if \lmysticeti fails to progress a round within $6\Delta$ \change{(as per Lemma~\ref{lem:guard-liveness}; this corresponds to GST in \lmysticeti's partial-synchrony model)}, then within \change{$2\Delta+$}$\Delta_{BA}$ thereafter, all honest \gossip validators agree on a \emph{liveness blameset} $\mathcal{B}$ with $|\mathcal{B}| \geq f_c+1$, where every $v \in \mathcal{B}$ is accompanied by sufficient evidence of non-responsiveness attested by a majority of guards.
\end{definition}
}

\subsection{\gossip Safety and Liveness}

\change{We now prove the main safety and liveness properties of \gossip. Proofs of supporting lemmas, \Cref{lem:guard-liveness}, and \Cref{thm:recovery} are deferred to \extref{Appendix E}{sec:guard-proofs}.}

\begin{lemma}[Safety Violations]\label{lem:guard-safety-bound}
    Any safety violation (equivocation) on \lmysticeti will be \change{detected and traced to a blameset of size $\ge f_c+1$} within $2\Delta+\Delta_{BA}$ after the first honest party observes it.
\end{lemma}
\begin{proof}
    Assume that an honest party (validator or guard) observes two equivocating blocks (according to its view) at time $t$. \change{The observer broadcasts the equivocating blocks to all guards via reliable broadcast. Within $\Delta$, all honest guards have received the equivocation; each independently applies \extref{Lemma~15}{lem:guard-safety5} to extract a valid safety blameset of size $\ge f_c+1$, then calls \textsc{Recover}, broadcasting its blameset to all guards. Within a further $\Delta$, all guards have collected each other's recovery messages. Thus by time $t+2\Delta$ all honest guards enter} $BA$. Within another $\Delta_{BA}$, all guards will agree on a blameset of at least $f_c+1$ validators \change{(at least one valid blameset exists as input, so BA produces a valid output). This establishes agreement on a valid safety blameset (and the recovery steps R1--R3 then restore the protocol).}
\end{proof}

\begin{lemma}[Liveness]\label{lem:guard-liveness}
    For every round $r$, \lmysticeti progresses from $r$ within $6\Delta$ of entering $r$, or at least $f_c+1$ parties are blamed in \gossip.
\end{lemma}

\change{%
\begin{theorem}[\gossip provides Accountable Safety]\label{thm:guard-accountable-safety}
    Under \gossip's synchrony assumption, \gossip satisfies \Cref{def:accountable-safety}: any safety violation in \lmysticeti results in all honest guards agreeing on a safety blameset of size $\geq f_c+1$ within $2\Delta+\Delta_{BA}$.
\end{theorem}
\begin{proof}
    By \extref{Lemma~15}{lem:guard-safety5}, any safety violation implies at least $f_c+1$ equivocating validators with cryptographic proof.
    \change{Each honest guard that observes the equivocation calls \textsc{Recover}, broadcasting its locally-constructed safety blameset; all guards collect these within $\Delta$ and use them as input to $BA$.}
    By Lemma~\ref{lem:guard-safety-bound}, the equivocation reaches all honest guards within $2\Delta$ via synchronous broadcast, and all guards run Byzantine Agreement on the resulting blameset, terminating within $\Delta_{BA}$.
    By \extref{Lemma~16}{lem:guard-honest-safe}, no honest validator appears in a valid safety blameset.
    The output of Byzantine Agreement is therefore a valid safety blameset of size $\geq f_c+1$.
\end{proof}

\begin{theorem}[\gossip provides Accountable Liveness]\label{thm:guard-accountable-liveness}
    Under \gossip's synchrony assumption, \gossip satisfies \Cref{def:accountable-liveness}: any liveness violation in \lmysticeti results in all honest guards agreeing on a liveness blameset of size $\geq f_c+1$ within \change{$2\Delta+$}$\Delta_{BA}$ after the liveness failure is confirmed.
\end{theorem}
\begin{proof}
    By Lemma~\ref{lem:guard-liveness}, if \lmysticeti fails to progress round $r$ within $6\Delta$, then at least $f_c+1$ core validators are in the intersection of all honest guards' $\asleep$ sets, and all honest guards hold $\geq S_f+1$ blame attestations for each such validator \change{(constituting a majority of guard stake, satisfying the $\Call{Maj}{\cdot}$ condition)}.
    By \extref{Lemma~17}{lem:guard-honest-live}, no honest validator is in any valid liveness blameset.
    \change{Each honest guard's $\lblamed(r)$ set therefore contains a valid liveness blameset of size $\ge f_c+1$; upon the $\graced$ timer firing, each guard calls \textsc{Recover}, which broadcasts its blameset via BC ($\Delta$), collects $\mathsf{RCVec}$ ($\Delta$), and then runs BA ($\Delta_{BA}$).}
    All honest guards thus have a valid liveness blameset of size $\geq f_c+1$ as input to Byzantine Agreement, which terminates within \change{$2\Delta+$}$\Delta_{BA}$.
\end{proof}
}

\change{%
\begin{theorem}[Recovery correctness]\label{thm:recovery}
    Under \gossip's synchrony assumption \change{(honest-majority BA over the guard set~\cite{abraham2020synchotstuff})}, following any safety or liveness violation of \lmysticeti:
    (i) all honest guards agree on a single blameset $\mathcal{B}$ with $|\mathcal{B}|\ge f_c+1$ containing only faulty validators;
    (ii) the remaining $4f_c$ core validators retain strict honest majority ($\ge 2f_c+1$) and resume consensus; and
    (iii) every state finalized by a checkpoint-path FinalityQC is preserved.
\end{theorem}
}

\section{Evaluation} \label{sec:evaluation}

\begin{figure*}[t]
    \vskip -1em
    \centering
    \captionsetup{aboveskip=4pt, belowskip=2pt}
    \begin{subfigure}{0.49\textwidth}
        \centering
        \includegraphics[width=0.9\linewidth]{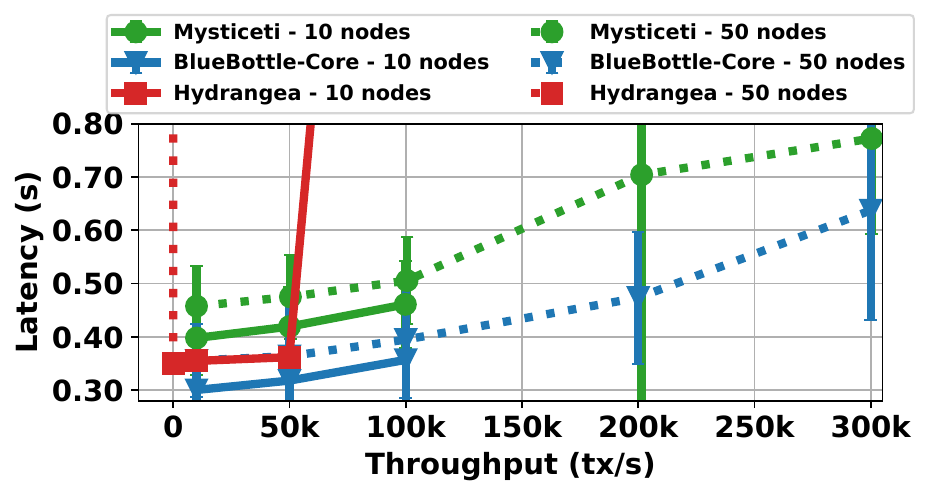}
        \caption{Committees with $10$ and $50$ validators, no validator faults.}
        \label{fig:evaluation-common}
    \end{subfigure}\hfill
    \begin{subfigure}{0.49\textwidth}
        \centering
        \includegraphics[width=0.9\linewidth]{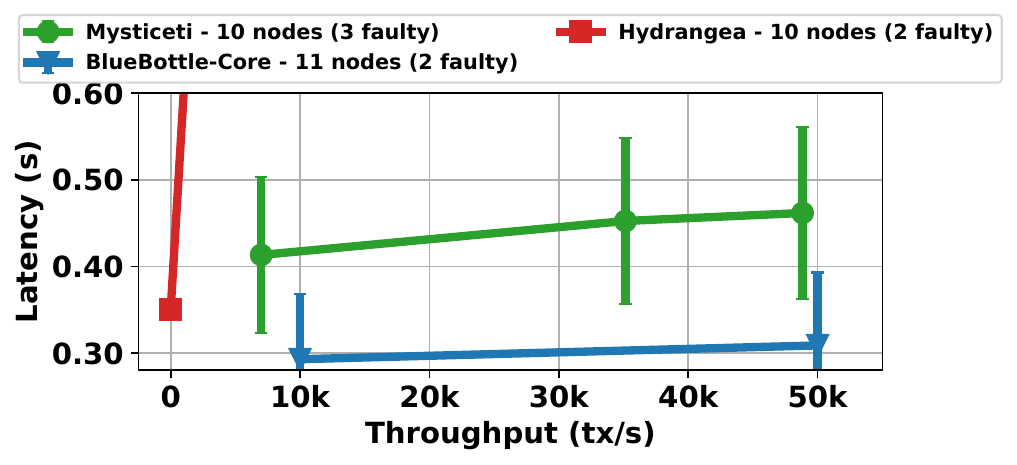}
        \caption{With $3$ or $2$ crash faults; committees with $10$ and $11$ validators (minimum to tolerate $3$ or $2$ faults).}
        \label{fig:evaluation-faults}
    \end{subfigure}
    \caption{ \footnotesize
        WAN throughput-latency performance comparison of \lmysticeti ($n=5f+1$) and Mysticeti ($n=3f+1$) with a $512$B transaction size. The y-axis starts at $300$\,ms to zoom in on the latency difference between the systems.
    }
    \label{fig:evaluation}
    \vskip -1em
\end{figure*}

We implement a \lmysticeti validator\footnote{\codelink}
in Rust by forking the Mysticeti codebase~\cite{mysticeti-code}. We provide more details about our implementation and testing methodology in \extref{Appendix F.1}{sec:implementation}.
We evaluate the throughput and latency of \lmysticeti through experiments conducted on Amazon Web Services (AWS) on a geo-distributed testbed.
We detail the experimental setup in \extref{Appendix F.2}{sec:experimental-setup}.

%

\para{Evaluation scope}
Our evaluation has one key goal: demonstrate that fault tolerance can be traded for lower latency. We do not claim that \lmysticeti is the best or most performant protocol in the $n=5f+1$ design space.
As a result, our primary baseline is Mysticeti~\cite{mysticeti}, one of the closest $n=3f+1$ consensus protocols to \lmysticeti in terms of design and implementation.

Specifically, we demonstrate the following claims:
\begin{enumerate}[label={\bf C\arabic*}]
    \item\label{claim:c1} \lmysticeti has similar throughput and lower latency than Mysticeti when operating in fault-free, synchronous networks.
    \item\label{claim:c2} \lmysticeti scales as well as Mysticeti and maintains high throughput and low latency as the number of validators increases.
    \item\label{claim:c3} \lmysticeti has a similar throughput to and lower latency than Mysticeti in the presence of (benign) crash faults.
\end{enumerate}

Note that evaluating the performance of BFT protocols in the presence of Byzantine faults is an open research question~\cite{twins}, and state-of-the-art evidence relies on formal proofs (presented in \extref{Appendix D}{sec:proofs}).

\para{Baseline protocols}
Recent work has proposed several consensus protocols that explore similar designs to the $n = 5f + 1$ setting (\Cref{sec:related}), including Minimmit~\cite{minimmit}, Kudzu~\cite{shoup2025kudzu}, Alpenglow~\cite{kniep2025alpenglow}.
None of these protocols provide deployable implementations with networking code; their codebases are intended for simulation only.
%
%

Hydrangea~\cite{shrestha2025hydrangea} operates with $n=3f+2c+k+1$, where $c$ and $k$ are additional parameters that allow trading fault tolerance for latency. Hydrangea can commit in a single round-trip and is, to our knowledge, the only protocol in this design space with a deployable implementation; we therefore include it as a secondary baseline.
Following the authors' specifications~\cite{shrestha2025hydrangea}, we set $f=1$, $c=2$, $k=2$ for experiments with 10 validators and $f=9$, $c=10$, $k=2$ for experiments with 50 validators.

\subsection{Benchmark under ideal conditions}
We evaluate the performance of \lmysticeti under normal, failure-free conditions in a wide-area network (WAN). \Cref{fig:evaluation} (Left) reports results from a geo-replicated deployment with a small committee of $10$ validators and a larger committee of $50$ validators. For cost reasons, we cap the input load at $100{,}000$\,tx/s ($10$ validators) and $300{,}000$\,tx/s ($50$ validators). These rates are two orders of magnitude above the observed peak throughput of any existing blockchain and we believe sufficient to assess system behavior under load.

The figure shows that \lmysticeti successfully trades some fault tolerance for lower latency. \lmysticeti lowers the commit path to two message delays. Each round of \lmysticeti must, however, wait for a larger parent quorum ($80\%$ of validators) than Mysticeti ($67\%$). Despite this larger quorum, \lmysticeti reduces end-to-end latency by about $20$--$25\%$ in the WAN setting across all input loads, confirming claim~\ref{claim:c1}. Concretely, with $10$ validators at $100{,}000$\,tx/s, \lmysticeti attains about $357$\,ms latency versus about $461$\,ms for Mysticeti; with $50$ validators at $100{,}000$\,tx/s, \lmysticeti reaches about $395$\,ms versus about $505$\,ms, and at $300{,}000$\,tx/s, \lmysticeti reaches about $637$~ms versus about $772$\,ms.
We observe a consistent proportional latency reduction at both committee sizes, confirming claim~\ref{claim:c2}.

Hydrangea's performance is significantly worse than both Mysticeti and \lmysticeti (red lines in \Cref{fig:evaluation}). With $10$ validators under no faults, Hydrangea sustains competitive latency (${\sim}355$\,ms) up to around $50{,}000$\,tx/s but degrades sharply at higher loads. With $50$ validators or under crash faults, Hydrangea's latency exceeds the visible scale of the plot even at minimal offered load, indicating that the protocol does not scale gracefully to larger committees in its current form. Hydrangea is a fork of HotStuff~\cite{hotstuff-code} that embeds data directly in the leader's proposal; its limited scalability is thus likely attributable to the data-dissemination mechanism inherited from HotStuff, consistent with the findings of Danezis et al.~\cite{narwhal}.

\subsection{Benchmark under faults}\label{sec:faults}
\Cref{fig:evaluation} (Right) shows \lmysticeti and Mysticeti under 3 crash faults in committees of 10 and 11 validators, (the minimum sizes used for this fault level in our experiments). We cap the offered load at $50{,}000$\,tx/s for cost reasons.
As expected, both systems sustain the input load with no material latency inflation relative to the fault-free runs. Mysticeti's latency lies between $400$ and $500$~ms, whereas \lmysticeti remains between $300$ and $350$\,ms, preserving a $20$--$25\%$ latency advantage (again at the cost of lowering fault tolerance) and confirming claim~\ref{claim:c3}. Both systems handle benign crashes gracefully by skipping faulty leaders via the direct skip rule described in \Cref{sec:lmysticeti}.
Hydrangea's latency again exceeds the visible scale of the plot, confirming its limited ability to handle faults gracefully, inline with the findings of Danezis et al.~\cite{narwhal}.

\section{Making the Core Asynchronous}

We study $\lmysticeti$ under partial synchrony and also introduce an asynchronous variant. The asynchronous algorithm mirrors the partially synchronous one, except for the orange-highlighted additions, with formal correctness proofs in \extref{Appendix G}{sec:async}. Using the asynchronous $\lmysticeti$ in $\sysname$ introduces a subtle effect: $\gossip$ must detect liveness faults via a different mechanism. 

In the partially synchronous design, liveness monitoring first checks whether the round leader is faulty and then verifies that the received blocks correctly reject a faulty leader.
\change{This procedure completes within the $6\Delta$ bound of \Cref{lem:guard-liveness}, where $\Delta$ denotes the network delay bound; of this, a $2\Delta$ component is the wait for the round's leaders. The asynchronous variant removes leader timeouts, eliminating this $2\Delta$ leader wait and reducing the liveness-detection bound to $4\Delta$.}
This yields two implications. The first one is straightforward: the partially synchronous protocol trades one round of optimistic-case latency for faster liveness recovery in the asynchronous protocol. The second one is more significant: the asynchronous variant removes the dependence on the network delay $\Delta$ between $\lmysticeti$ and $\gossip$. This decoupling lets operators choose a conservative $\Delta$ (even tens of seconds) without harming $\lmysticeti$’s performance under crash faults. Operators can adjust this value at every epoch. Moreover, a faulty-leader detection and exclusion mechanism, such as HammerHead~\cite{hammerhead}, further mitigates the impact of faulty leaders.

\section{Related Work} \label{sec:related}
Byzantine fault-tolerant state machine replication (BFT-SMR) is a foundational abstraction in distributed systems, and its latency limits have been studied extensively~\cite{martin2005fast,abraham2020synchotstuff, abraham2020byzantine,pbft,gueta2019sbft,kotla2007zyzzyva}. Classical work emphasizes worst-case latency, with state-of-the-art protocols committing in three communication steps while tolerating up to $f < n/3$ Byzantine faults~\cite{pbft,buchman2019latest,chan2023simplex}.

Motivated by practical deployments (e.g., blockchains), recent work increasingly targets good-case latency under partial synchrony, i.e., the latency to commit when the designated leader is correct and the network satisfies partial synchrony assumptions~\cite{abraham2022goodcase,kuznetsov2021revisiting}. In this setting, protocols aspire to two-round commitment. SBFT~\cite{gueta2019sbft} achieves two rounds in the absence of Byzantine faults and reverts to a slower path otherwise. FAB~\cite{martin2005fast} removes this restriction but was later shown to suffer from a liveness issue~\cite{abraham2018revisiting}. Kudzu~\cite{shoup2025kudzu} commits in three rounds in general and in two rounds when the number of Byzantine faults is small ($f < n/5$). Minimmit~\cite{minimmit} is a recently proposed $2$-round variant of Simplex~\cite{chan2023simplex,shoup2024song,decentralized-5f-1} with $n=5f+1$, and reports, similarly to Kudzu, latency results only through simulation based on the protocol's finalization rule, i.e., the number of communication steps and threshold required before a block can be committed. Minimmit and Kudzu's fast path use the same finalization structure, requiring one round of communication with a $4f+1$ threshold, and therefore produce identical block latency (see Tables 3, 5, 6, 7 of \cite{minimmit}). Hydrangea~\cite{shrestha2025hydrangea} adopts a generalized fault model that distinguishes Byzantine from crash faults; while this model is weaker than a purely Byzantine model for a fixed total fault budget, it requires careful parameterization for practical deployment. Alpenglow~\cite{kniep2025alpenglow} claims two-round commitment while simultaneously tolerating less than $20\%$ Byzantine and less than $20\%$ crash faults; however, this guarantee this is impossible according to~\cite{shrestha2025hydrangea} and after careful analysis it is actually a 3$\delta$ protocol as Rotor requires 2$\delta$ and Votor an extra one.
Optimistic fast paths have also been explored in synchronous and asynchronous settings~\cite{aublin2015next,brasileiro2001consensus,friedman2005simple,kursawe2002optimistic,pass2018thunderella,shrestha2020optimality,shrestha2025optimistic}.

For comparison, \sysname commits in two rounds under partial synchrony while tolerating up to $f_c$ Byzantine faults with $n_c = 5f_c+1$ replicas. Two-round finality with $n = 5f-1$ replicas is achievable when the protocol explicitly identifies and ignores faulty leaders~\cite{kuznetsov2021revisiting,decentralized-5f-1}. \sysname could be extended with such leader-exclusion mechanisms to attain the $5f-1$ bound; however, we deliberately avoid this design point because the incremental gain in fault tolerance is modest relative to the additional protocol design and engineering complexity. \change{FlexibleBFT~\cite{malkhi2019flexible} and Optimal Flexible Consensus~\cite{neu2024optimal} formalize heterogeneous client finality, giving clients with different fault-tolerance assumptions optimal safety/liveness guarantees; \sysname's fast and checkpoint paths serve the same client-choice goal, while \gossip adds the reconfiguration guarantee that flexible-consensus constructions do not address.}

Accountable safety~\cite{buterin2019casper,neu2022availability,sheng2021bft, neu2023accountable} and accountable liveness~\cite{lewispye2025accountable,tas2023bitcoin} strengthen classical safety and liveness via crypto-economic mechanisms. In addition to preserving agreement (safety) and eventual decision (liveness) below a fault threshold, these notions enable attribution of provable misbehavior, e.g., through slashing conditions, when safety or liveness is violated, thereby deterring equivocation or censorship in permissioned and permissionless settings. \change{\gossip instantiates these notions as an accountability gadget~\cite{buterin2019casper} over \lmysticeti: a synchronous overlay that monitors the core protocol and, upon detecting a violation, identifies a provable blameset and reconfigures the validator set. Recent work on recovery from excessive faults~\cite{civit2025recover,gong2025recovery,sheng2023bft} addresses how protocols can resume once the corruption threshold is exceeded; unlike these works, \gossip actively reduces the effective corruption count by removing provably faulty validators before resuming, tying fork-choice to the FinalityQC chain for an unambiguous recovery rule.}

\ifpublish
\begin{acks}
    This work is partially supported by Mysten Labs.
\end{acks}

\fi

\bibliographystyle{ACM-Reference-Format}
\bibliography{references}

\iflong
\appendix

\section{Open Science}
\label{sec:open-science}

We are committed to open science.
All artifacts required to reproduce the results of this paper are available at anonymous hosting services
for double-blind review:
\begin{itemize}
  \item \textbf{Core protocol implementation} (\lmysticeti): \codelink
  \item \textbf{Asynchronous variant} (\asynclmysticeti): \asynccodelink
  \item \textbf{Experimental dashboard and raw data}: \dashboardlink
\end{itemize}
The repositories include the source code, build scripts, orchestration tooling, and configuration files used to run the experiments reported
in \Cref{sec:evaluation}, along with documentation for reproducing the
benchmarks.
No artifact is withheld.
We provide a tutorial for reproducing the experiments in \Cref{sec:tutorial}.
Upon acceptance, we will de-anonymize the repositories and release a tagged artifact for the ACM Artifact Evaluation process.

\section{Generative AI Usage}
\label{sec:gen_ai}

Various AI tools, such as Claude, Gemini, ChatGPT/Codex, and Grammarly, were used to assist with spell-checking, grammar, and general English writing.
All changes made by the AI model were reviewed and approved by the authors.
All content of the paper was written by the authors themselves, and the authors retain full responsibility for the accuracy, originality, and integrity of the work.

\section{Ethical Considerations}
\label{sec:ethics}

This work is a systems and protocol paper that does not involve human subjects, user data, personally identifiable information, or the exploitation of deployed systems.
All experiments are run on validator nodes under our own control (cloud VMs) and do not interact with any third-party or production network.
The protocols described here target permissionless and permissioned blockchain deployments where safety and liveness are the explicit goals; we see no realistic dual-use concerns beyond those common to consensus research.
No responsible disclosure process was required.

\section{\lmysticeti Safety and Liveness} \label{sec:proofs}

We prove that \lmysticeti implements Byzantine Atomic Broadcast (BAB)~\cite{cristian1995}, which formalizes reliable message delivery with total ordering in the presence of Byzantine faults. In BAB, \emph{broadcast} refers to a validator proposing a message for inclusion in the final sequence, while \emph{deliver} refers to outputting a message from the total ordering for local processing. A BAB protocol must satisfy four properties:

\begin{itemize}
    \item \textbf{Validity:} If an honest validator broadcasts a message $m$, then all honest validators eventually deliver $m$.
    \item \textbf{Agreement:} If an honest validator delivers a message $m$, then all honest validators eventually deliver $m$.
    \item \textbf{Integrity:} For any message $m$, an honest validator delivers $m$ at most once, and only if $m$ was previously broadcast by some validator.
    \item \textbf{Total Order:} If two honest validators both deliver messages $m_1$ and $m_2$, they deliver them in the same order.
\end{itemize}

\begin{lemma}\label{lem:safety1}
    There will never be a block that an honest validator directly commits while another honest validator directly skips.
\end{lemma}
\begin{proof}
    Assume for the sake of contradiction that such a block exists and call it $B$. Thus, there are $4f_c + 1$ validators which support $B$ and $4f_c + 1$ validators which blame $B$. Since $f_c$ validators are Byzantine, there are $3f_c + 1$ honest validators that support $B$ and a distinct set of $3f_c + 1$ honest validators that blame $B$. This means there are $7f_c + 2$ validators overall contradicting our assumption of $n_c=5f_c+1$.
\end{proof}

\begin{lemma}\label{lem:safety2}
    There will never be a block that an honest validator directly skips while another honest validator indirectly commits.
\end{lemma}
\begin{proof}
    Assume for the sake of contradiction that such a block exists and call it $B$. Thus, there are $4f_c + 1$ validators which blame $B$ and $2f_c + 1$ validators which support $B$. Since $f_c$ validators in the network are Byzantine, there are $3f_c + 1$ honest validators that blame $B$ and a distinct set of $f_c + 1$ honest validators that support $B$. This means there are $5f_c + 2$ validators overall contradicting our assumption of $n_c=5f_c+1$.
\end{proof}

\begin{lemma}\label{lem:safety3}
    If at a round $R$, $4f_c + 1$ blocks from distinct validators support a block $B$, then all blocks at future rounds $R' > R$ will link to $2f_c + 1$ supports for $B$ from round $R$.
\end{lemma}
\begin{proof}
    Each block links to $4f_c + 1$ blocks from the previous round. For the sake of contradiction, assume that a block $B'$ in round $R' > R$ does not link to $2f_c + 1$ supports for $B$ from round $R$.
    \case{$R' = R + 1$. $B'$ should link to $4f_c + 1$ blocks from round $R$ and also $4f_c + 1$ blocks in round $R$ should support $B$. By quorum intersection, any two subsets of size $4f_c+1$ from a set of size $5f_c+1$ must share at least $3f_c+1$ elements. Since $f_c$ validators are Byzantine, the intersection contains at least $2f_c+1$ correct validators. Thus, the only way for $B'$ to not link to $2f_c + 1$ supports is if an honest validator equivocated in round $R$. This is a contradiction.}
    \case{$R' > R + 1$. $B'$ links to $4f_c + 1$ blocks from round $R' - 1$. At least $3f_c + 1$ of these blocks are produced by honest validators. Honest validators always link to their own blocks, which means they will eventually link to their block from round $R + 1$. The above case proves how these blocks from round $R + 1$ link to $2f_c + 1$ supports for $B$. Thus, the only way for $B'$ to not link to $2f_c + 1$ supports for $B$ is if all of these honest validators do not link to their own block in round $R + 1$. This is a contradiction.}
\end{proof}

As a result of Lemma \ref{lem:safety3}, we have the following corollary.

\begin{corollary}\label{cor:safety4}
    There will never be a block which an honest validator directly commits while another honest validator indirectly skips.
\end{corollary}

\begin{lemma}\label{lem:safety5}
    All honest validators who have decided on a leader block, agree on the decision.
\end{lemma}
\begin{proof}
    Let $B_j$ and $B_k$ be the highest committed leader blocks according to validators $X$ and $Y$ respectively. Without loss of generality, let $j \leq k$. Note that leader blocks decided by $X$ which are higher than $B_j$ are direct skips which according to Lemma \ref{lem:safety1} and Lemma \ref{lem:safety2} will be consistent with $Y$'s decision. The proof continues by induction on the statement for $0 \leq i \leq j$, if both $X$ and $Y$ decide on leader block $B_i$, then they either both commit or both skip the block.
    \case{$i = j$. By definition, $X$ directly commits $B_i$ and from Lemma \ref{lem:safety1} and Corollary \ref{cor:safety4}, $Y$ will also commit $B_i$.}
    \case{Assuming the statement is true regarding $B_i$ for $\ell + 1 \leq i \leq j$, we prove it is true for $B_\ell$. This is done by enumerating decision possibilities.
    \begin{enumerate}
        \item If either validator directly commits $B_\ell$, then by Lemma \ref{lem:safety1} and Corollary \ref{cor:safety4}, the other will commit.
        \item If either validator directly skips $B_\ell$, then by Lemma \ref{lem:safety1} and Lemma \ref{lem:safety2}, the other will skip.
        \item Both $X$ and $Y$ indirectly decide $B_\ell$. Let $A_c^X$ and $A_d^Y$ be the anchors used by $X$ and $Y$ to indirectly decide $B_\ell$. Since $\ell + 1 < c \leq j$, it follows from the induction hypothesis that $A_c^X = A_d^Y$. Thus, both $X$ and $Y$ use the same anchor to decide $B_\ell$. The indirect decision rule solely depends on the causal history of the anchor. By using the same anchor, $X$ and $Y$ will agree on the decision for $B_\ell$.
    \end{enumerate}}
\end{proof}

\begin{theorem}[Total Order]\label{thm:safety}
    \lmysticeti satisfies the total order property of BAB.
\end{theorem}
\begin{proof}
    By Lemma~\ref{lem:safety5}, all honest validators agree on the decision for every leader block. The total ordering is a deterministic function of the sequence of committed leader blocks. Since all honest validators share the same sequence of committed leader blocks, they produce identical total orderings and deliver all blocks in the same order.
\end{proof}

\begin{theorem}[Integrity]\label{thm:sync-integrity}
    \lmysticeti satisfies the integrity property of BAB.
\end{theorem}
\begin{proof}
    A block enters an honest validator's DAG only if it carries a valid signature from its author, so no unbroadcast block is ever delivered. The deterministic linearization processes each block as part of exactly one committed leader's causal history, so no honest validator delivers the same block twice.
\end{proof}

\begin{lemma}\label{lem:liveness1}
    After GST, all honest validators will enter the same round within $\Delta$.
\end{lemma}
\begin{proof}
    Messages sent before GST will deliver in $\Delta$ after GST commences. Thus, the valid block of the highest round that any validator sent before GST will be delivered to all validators in GST + $\Delta$. Upon receiving this block, all honest validators will enter the round.
\end{proof}

\begin{lemma}\label{lem:liveness2}
    After GST, leader blocks from honest validators will receive support from all honest validators.
 \end{lemma}
 \begin{proof}
    By Lemma \ref{lem:liveness1}, all honest validators will enter the same round within $\Delta$ after GST. When an honest validator sends its leader block for this round, it is delivered to all validators within $\Delta$. Since the protocol sets the timeout to $2 \Delta$, any honest validator that has entered the round will receive the leader block and cast support for it before timing out, regardless of when exactly they entered the round relative to other honest validators.
 \end{proof}

As a result of Lemma \ref{lem:liveness2}, we have the following corollary. Recall that there are $4f_c + 1$ honest validators.

\begin{corollary}\label{cor:liveness3}
    After GST, leader blocks from honest validators will be (directly) committed.
\end{corollary}

\begin{lemma}\label{lem:liveness4}
    The round-robin schedule of leader-block proposers ensures that, within a window of $2f_c + 2$ rounds, there are two consecutive rounds in which an honest validator is the proposer of the highest-ranked leader block.
\end{lemma}
\begin{proof}
    The network contains $f_c$ Byzantine validators. In $2f_c + 2$ rounds, there are $2f_c + 1$ sets of two consecutive rounds. Due to the schedule being round robin, in at least $f_c + 1$ of the rounds, an honest validator will be the proposer of the highest-ranked leader block. These blocks are the highest-ranked leader block in exactly two of the sets. By the pigeonhole principle, one set must contain $\lceil \frac{2 \cdot (f_c + 1)}{2f_c+1} \rceil = 2$ honest validators proposing the highest-ranked leader block.
\end{proof}

\begin{lemma}\label{lem:liveness5}
    After GST, undecided leader blocks will eventually be decided.
\end{lemma}
\begin{proof}
    Let $B$ be an undecided leader block in round $r$. By Lemma \ref{lem:liveness4}, after GST, there will be two consecutive rounds, $j$ and $j+1$ with $j > r$, where honest validators propose the highest-ranked leader block. By Corollary \ref{cor:liveness3}, their leader blocks will be committed. The proof continues by induction on the statement for rounds earlier than $j$, all leader blocks are decided.
    \case{All undecided leader blocks in rounds $j-1$ and $j-2$ will be decided as they now have decided anchors in rounds $j+1$ and $j$ respectively.}
    \case{For undecided leader blocks in round $i < j-2$, $j$ is higher than the decision round of the wave that $i$ is in. Thus, by the induction hypothesis, there are no undecided leader blocks between $i$ and $j$. Hence, the leader block in round $i$ will also be decided.}
\end{proof}

\begin{theorem}[Validity]\label{thm:liveness}
    \lmysticeti satisfies the validity property of BAB.
\end{theorem}
\begin{proof}
    By Lemma~\ref{lem:liveness1}, all honest validators synchronize to the same round within $\Delta$ after GST. By Lemma~\ref{lem:liveness4}, within any window of $2f_c + 2$ rounds, there exist two consecutive rounds where honest validators propose the highest-ranked leader blocks. By Corollary~\ref{cor:liveness3}, these honest leader blocks will be directly committed. By Lemma~\ref{lem:liveness5}, any previously undecided leader blocks will eventually be decided once we have committed blocks from honest validators in consecutive rounds. Since honest leader blocks are committed every $2f_c + 2$ rounds and each commitment resolves all pending undecided blocks, no block remains undelivered indefinitely.
\end{proof}

\begin{theorem}[Agreement]\label{thm:agreement}
    \lmysticeti satisfies the agreement property of BAB.
\end{theorem}
\begin{proof}
    Suppose an honest validator delivers a block $B$. Then $B$ was delivered as part of the causal history of some committed leader block $B'$. By Lemma~\ref{lem:safety5}, all honest validators agree on the decision for $B'$. By Theorem~\ref{thm:liveness}, all honest validators eventually decide every leader block. Therefore, all honest validators eventually commit $B'$ and deliver $B$.
\end{proof}

\subsection{Checkpoint Safety} \label{sec:checkpoint-proofs}

We now prove the safety of the checkpoint path under the extended threat model that includes alive-but-corrupt (AbC) validators.

\begin{lemma}[Unique CheckpointQC per height]\label{lem:unique-checkpoint}
    Let $n_c = 5f_c+1$ and $q = 4f_c+1$. If the total number of equivocating validators (Byzantine plus AbC) satisfies $f_c + \fabc \leq 3f_c$, then no two conflicting CheckpointQCs can be formed for the same slot height~$s$.
\end{lemma}
\begin{proof}
    By the protocol (\Cref{sec:checkpoints}), an honest validator broadcasts at most one checkpoint proposal per slot height.
    Suppose, for contradiction, that two CheckpointQCs $C_1$ and $C_2$ exist for the same height $s$ but for different state roots ($\sigma_1 \neq \sigma_2$).
    Each requires $q = 4f_c+1$ signatures from distinct validators.
    Since honest validators do not sign conflicting checkpoint proposals at the same height, any validator appearing in both signer sets must be equivocating.
    The minimum overlap between the two signer sets is $2q - n_c = 2(4f_c+1) - (5f_c+1) = 3f_c+1$.
    Thus, at least $3f_c+1$ validators must equivocate.
    This contradicts our assumption that the total number of equivocating validators is at most $f_c + \fabc \leq 3f_c < 3f_c+1$.
\end{proof}

\begin{theorem}[Checkpoint-path safety]\label{thm:resilient-safety}
    If two clients finalize heights via the checkpoint path by observing FinalityQCs, they will never commit conflicting states for the same slot height, provided $f_c + \fabc \leq 3f_c$.
\end{theorem}
\begin{proof}
    A client finalizes a state for height $s$ only upon observing a FinalityQC, which requires an underlying valid CheckpointQC.
    By \Cref{lem:unique-checkpoint}, there is at most one CheckpointQC per height $s$.
    Since honest validators only broadcast a checkpoint witness matching the uniquely determined CheckpointQC (and only if its payload executes to the correct state), no FinalityQC can be formed for a conflicting state.

    During fork recovery (\Cref{sec:gossip}), validators exchange all known CheckpointQCs via Byzantine Broadcast.
    Because \Cref{lem:unique-checkpoint} guarantees at most one CheckpointQC exists per height, the sequence of CheckpointQCs forms a single non-conflicting chain.
    All correct validators deterministically adopt the branch belonging to the CheckpointQC with the highest slot, preserving any state finalized by a checkpoint-path client prior to the fork.
\end{proof}

\section{\gossip Safety and Liveness}\label{sec:guard-proofs}

\noindent Intermediate results and all proofs from \Cref{sec:gossip} appear here.
We first state and prove the supporting lemmas not restated in the main text, then give full proofs of the main-text results.
We also include the full recovery-protocol details (Steps R1--R3) and a comparison with related recovery work.

\begin{lemma}\label{lem:guard-safety}
Assume there are $n_c = 5f_c+1$ core validators and up to $3f_c$ of them are malicious. No two honest validators can directly commit different leader blocks (or directly commit and directly skip) for the same leader and round.
\end{lemma}
\begin{proof}
    Assume that this can occur, and let $B_{(l,r)}\ne B'_{(l,r)}$ be the two distinct leader blocks that the two honest validators respectively commit (or directly skip) for the same leader $l$ and round $r$.
    Then, there are $4f_c + 1$ validators which vote $B_{(l,r)}$ and $4f_c + 1$ validators which vote (or skip) $B'_{(l,r)}$.
    Since there are $3f_c$ malicious validators, there are still $f_c+1$ honest votes needed for $B_{(l,r)}$ and another $f_c+1$ honest votes needed for $B'_{(l,r)}$ (or to skip $B'$). \change{Since honest validators never cast conflicting votes, these must be $2(f_c+1)=2f_c+2$ distinct honest validators. But only $2f_c+1$ honest validators exist; contradiction.}
\end{proof}

\begin{corollary}\label{cor:guard-safety1}
    The safety guarantee of \Cref{lem:safety1} holds even under $3f_c$ corruptions.
\end{corollary}
\begin{proof}
   \change{Under up to $3f_c$ corruptions, a direct commit and a direct skip of the same block each require $4f_c+1$ votes. With $3f_c$ Byzantine validators, each quorum must contain at least $f_c+1$ honest validators. Since no honest validator votes for both sides, we need $2(f_c+1)=2f_c+2$ distinct honest validators, but only $5f_c+1-3f_c=2f_c+1$ honest validators exist; contradiction.}
   \change{Equivalently, any two $q=4f_c{+}1$ quorums over $n_c=5f_c{+}1$ validators intersect in at least $2(4f_c{+}1)-(5f_c{+}1)=3f_c+1$ parties; forming two conflicting quorums would therefore require at least $3f_c+1$ equivocators, exceeding the $3f_c$ corruption bound;  contradiction.}
\end{proof}

\begin{lemma}\label{lem:guard-safety4}
    If there exists a block which an honest core validator directly commits while another honest core validator indirectly skips, then there exist at least $f_c+1$ validators that have provably equivocated.
\end{lemma}
\begin{proof}
    Assume that such a block exists and call it $B$, let $X$ be the honest validator that directly commits $B$, and let $Y$ be the honest validator that indirectly skips $B$'s slot.
    Each block links to $4f_c + 1$ blocks from the previous round.
    \change{Since $B$ is directly committed, exactly $4f_c+1$ round-$r$ blocks support $B$. For $Y$'s anchor to indirectly skip $B$, its causal history from round $r$ must include fewer than $2f_c+1$ B-supporting blocks. Therefore there must exist some round-$(r{+}1)$ ancestor $B'$ of the anchor whose $4f_c+1$ round-$r$ parents contain fewer than $2f_c+1$ B-supports. (If every round-$(r{+}1)$ ancestor of the anchor carried $\ge 2f_c+1$ B-supports from round $r$, the anchor's aggregate causal history would also contain $\ge 2f_c+1$ B-supporters, contradicting the indirect skip.)}
    $B'$ links to $4f_c + 1$ blocks from round $r$. Since $4f_c + 1$ blocks in round $r$ support $B$, and the total of distinct core validators is $5f_c+1$, \change{the two quorums of size $4f_c+1$ overlap in at least $2(4f_c+1)-(5f_c+1)=3f_c+1$ validators; for $B'$ to carry fewer than $2f_c+1$ supports of $B$, at least $(3f_c+1)-2f_c = f_c+1$ of these validators must have issued a round-$r$ block supporting $B$ and a distinct round-$r$ block (in $B'$'s parents) not supporting it.} Hence the minimum number of equivocations is $f_c+1$, and it is provable by their different votes.
\end{proof}

\begin{lemma}\label{lem:guard-safety2}
    If there exists a block that an honest core validator directly skips while another honest core validator indirectly commits, then there exist at least $f_c+1$ validators that have provably equivocated.
\end{lemma}
\begin{proof}
    Assume that such a block exists and call it $B$. \change{In BB-Core, a direct skip of $B$ requires $4f_c+1$ validators to have cast explicit skip (blame) votes against $B$; an indirect commit of $B$ requires $2f_c+1$ validators to have supported $B$ (the weak certificate).} Thus, there are $4f_c + 1$ validators which blame $B$ and $2f_c + 1$ validators which support $B$.
    This requires a total of $6f_c+2$ votes, out of which only $5f_c+1$ can be distinct.
    As a result, $6f_c+2 - (5f_c+1) = f_c+1$ \change{validators each cast both a blame vote and a support vote for $B$, constituting a cryptographic proof of equivocation.}
\end{proof}

\begin{lemma}\label{lem:guard-safety5}
    \change{Assume there are $n_c=5f_c+1$ core validators and up to $3f_c$ of them are malicious.} If there exist two honest core validators who have decided on a leader block and do not agree on the decision, then there exist at least $f_c+1$ core validators that have provably equivocated.
\end{lemma}
\begin{proof}
    Let $B_X\neq B_Y$ be the committed (or skipped) leader blocks according to validators $X$ and $Y$ respectively for leader $l$ in round $r$.
    By Lemma~\ref{lem:guard-safety}, no two honest validators could have both directly committed conflicting blocks. Similarly, no two honest validators could have directly committed and directly skipped a block respectively (Corollary~\ref{cor:guard-safety1}). \change{Any combination in which both validators skip (directly or indirectly) implies agreement, hence the remaining cases exhaustively are:}

    \case{$X$ directly committed $B_X$, while $Y$ \change{indirectly} skipped \change{the slot} (and vice versa, i.e., $X$ directly skipped while $Y$ indirectly committed): \change{The direct-commit/direct-skip combination is already excluded by Lemma~\ref{lem:guard-safety} and Corollary~\ref{cor:guard-safety1}.} From Lemma~\ref{lem:guard-safety4} (and Lemma~\ref{lem:guard-safety2} respectively), there will exist a set of at least $f_c+1$ provably misbehaving parties.}
    \case{$X$ directly committed $B_X$, while $Y$ indirectly committed $B_Y$ (and vice versa): Similar to previous arguments, this requires at least $(4f_c+1)+(2f_c+1)$ distinct votes, i.e., at least $f_c+1$ provable equivocations. \change{Explicitly: $4f_c+1$ validators support $B_X$ (direct commit) and $2f_c+1$ support $B_Y$ (weak certificate), with $B_X\ne B_Y$ at the same leader slot. Since $6f_c+2>(5f_c+1)$, at least $f_c+1$ validators supported both, constituting equivocation proofs (same counting as Lemma~\ref{lem:guard-safety2}).}}
    \case{$X$ indirectly committed $B_X$, while $Y$ indirectly committed $B_Y$: \change{We consider three sub-cases.} Let $A^X$ and $A^Y$ be the directly-committed anchors used by $X$ and $Y$ to indirectly commit $B_X$ and $B_Y$, respectively. 
    \change{(a) If $A^X=A^Y$, the indirect decision rule depends only on the anchor's causal history and is deterministic, so it would yield $B_X=B_Y$; contradiction. 
    (b) If $A^X$ and $A^Y$ are conflicting blocks of the \emph{same} slot, then} since no honest party could have voted for both anchors, there are not enough votes for both anchors to be directly committed (total votes needed $8f_c+2$, but \change{the maximum number of votes available is $2f_c+1+6f_c=8f_c+1$, since the $2f_c+1$ honest validators each vote for at most one of the two conflicting anchors while the $3f_c$ malicious can vote for both). 
    This sub-case is therefore impossible and no further argument is needed. 
    (c) Finally, if $A^X$ and $A^Y$ lie in different rounds, say $A^X$ is the lower one, then $Y$ must have skipped the slot of $A^X$ in order to select a higher anchor while $X$ directly committed it. This is a disagreement on $A^X$'s leader slot where $X$ directly committed and $Y$ made a different decision (indirect skip or indirect commit of that slot), 
    which by Cases 1 or 2 of this proof exposes $\ge f_c+1$ equivocations.} Applying this reduction recursively always terminates at a direct disagreement. \change{Termination is guaranteed because each step moves to a strictly lower anchor slot; since leader slots are well-ordered, the recursion bottoms out within finitely many steps.}}
    \case{$X$ indirectly committed $B_X$, while $Y$ indirectly skipped \change{the slot. We consider three sub-cases.} In that case, let $A$ denote the directly-committed anchor used to indirectly commit $B_X$, and $A'$ denote the directly-committed anchor used  to indirectly skip $B_Y$. \change{\change{(a)} If $A=A'$, the same anchor cannot both reference and not reference a weak certificate for the slot, so $X$ and $Y$ would reach the same decision;  contradiction. \change{(b)} If $A$ and $A'$ are conflicting anchors of the same slot, then} since no honest party could have voted for both anchors, there are not enough votes for both anchors to be directly committed \change{(the same $8f_c+1<8f_c+2$ counting argument as in the previous case).} \change{This sub-case is therefore impossible.} \change{\change{(c)} If $A$ and $A'$ are at different rounds, the lower anchor is directly committed by one validator and \change{a different decision is made} by the other (indirect skip or indirect commit of that slot) \change{in reaching the higher anchor.}} \change{Termination follows by the same well-ordering argument as in the previous case.}}
\end{proof}

\begin{lemma}\label{lem:guard-honest-safe}
    No honest core validators will ever be included in a valid safety-blameset.
\end{lemma}
\begin{proof}
    Honest validators receive blocks and vote based on whether the blocks preserve the rules of the protocol. \change{Honest validators never sign two conflicting messages for the same slot; consequently, no valid equivocation proof---which requires two distinct signed messages from the same validator for the same slot---can implicate an honest validator. Therefore no honest validator can be included in any valid safety blameset.}
\end{proof}

\begin{lemma}\label{lem:guard-honest-live}
    No honest core validators will ever be included in a valid liveness-blameset.
\end{lemma}
\begin{proof}
    \change{Assume that an honest core validator $v$ is correctly included in a liveness-blame for some round $r$, i.e., a majority of guard stake (more than $S/2$) has attested to blaming $v$ for round $r$, satisfying the $\Call{Maj}{\cdot}$ threshold of Algorithm~\ref{alg:guard}. Since corrupted stake is at most $S_f<S/2$, the attesting majority must include at least one honest guard; call it $g$.}
    Then, according to Algorithm~\ref{alg:guard}, either i) $v$ did not vote for a leader that $g$ considers live, or ii) $v$ did not vote on time, or iii) $v$ was a leader and did not send on time.

    For the first case, $g$ considers leaders on time, only if it receives a valid leader block for the previous round within $2\Delta$ after the round has started. Even if the previous round progressed instantly, within $\Delta$, $v$ must have also progressed and received the leader's block.
    If $v$ then voted for the leader, within another $\Delta$, $v$'s block voting for the leader would arrive to $g$, on time (within the $4\Delta$ duration of $\lived^g$).
    As such, $v$ must have seen a leader block on time, and still not voted for it, which contradicts the honest behavior.

    Similarly for the second case, $g$ considers on time any valid vote that arrives within $4\Delta$ after $g$ enters the round. Since $v$ will enter the round by at most a $\Delta$ delay after $g$, if $v$ were to follow the protocol, it would wait by at most another $2\Delta$ for the leader's block and then submit its vote. Since it takes at most an additional $\Delta$ for $v$'s vote to arrive at $g$, $v$ must have entered the round on time, and still not voted for it on time, which contradicts the honest behavior.

    Finally, if $v$ was a leader for a round, it enters the round within at most $\Delta$ after the first honest party. An honest leader proposes a block upon entering a round for which it is the leader. Since $g$ would wait for $2\Delta$ for the leader's block after entering the round, even if $v$ entered the round $\Delta$ time after $g$, it would still be on time to send its block to $g$. Thus, $v$ must have entered the round on time, but not have sent its block on time, which contradicts honest behavior.
\end{proof}

\begin{corollary}\label{cor:no-recover}
    If the number of corruptions is $\le f_c$, then no Recover call occurs.
\end{corollary}
\begin{proof}
    \change{A valid blameset requires $\ge f_c+1$ provably misbehaving validators. With at most $f_c$ corruptions no honest validator appears in a safety blameset (Lemma~\ref{lem:guard-honest-safe}) or a liveness blameset (Lemma~\ref{lem:guard-honest-live}), so at most $f_c$ validators can be blamed. Hence no valid blameset of size $\ge f_c+1$ can be constructed, and no call to \textsc{Recover} occurs.}
\end{proof}

\begin{proof}[Proof of \Cref{lem:guard-liveness}]
    First, parties \emph{enter} round $r$ of \lmysticeti, once they receive $4f_c+1$ blocks (from distinct senders) from round $r-1$.
    Every guard will also enter rounds according to this logic.
    From synchrony and broadcast, it is guaranteed that once a guard enters round $r$, every participant will enter round $r$ at most within $\Delta$ as well (similar to \Cref{lem:liveness1}).
    Each guard sets a timer $\lived = 4\Delta$, by which it expects to receive blocks from at least $4f_c+1$ core validators; this is because in the worst-case, it would take at most an additional $\Delta$ for the slowest validator to enter the round, it would wait for the leaders of the past round for $2\Delta$ and would take an additional $\Delta$ for its block to return to the guard.
    Each guard expects to have received blocks from the leaders of the past round already before the current round starts, or within $2\Delta$ of the current round start; this is because in the worst-case, it would take an additional $\Delta$ for the slowest leader to enter the past round, and it would take another $\Delta$ for its block to arrive at the guard.

    Say $g^*$ is the first honest guard that enters round $r$, and let $\asleep_{g^*}(r)$ denote the set once $\lived^{g^*}=4\Delta$ expires.
    Within another $\Delta$ after $\lived^{g^*}$ ($5\Delta$ total), every other guard's (say $g$) timer will also expire and they will be blaming each of the asleep core validators $v\in\asleep_{g}(r)$.
    If $\lvert\asleep_{g}(r)\rvert\le f_c$ for any $g\in H$, with $H$ denoting the set of all honest guards, then it is guaranteed that all guards can progress the round after another $\Delta$ ($6\Delta$ total). \change{Specifically: since $|\asleep_g(r)|\le f_c$, guard $g$ received valid round-$r$ blocks from $\ge 4f_c+1$ validators by the time $\lived^g$ expired (at most $5\Delta$ after $g^*$ entered $r$). By synchrony, all other honest guards receive these same blocks within $\Delta$, giving each the $4f_c+1$-block quorum needed to advance to round $r+1$ by $6\Delta$.}
    The same holds if $\lvert\cap_{g\in H} \asleep_g(r)\rvert \le f_c$, since locally, every guard by that time will have blocks from all $v\notin\cap_{g\in H} \asleep_g(r)$, which are sufficient to progress the \lmysticeti round.
    Otherwise, by $6\Delta$ total, all guards will have at least $\lvert H\rvert = S_f+1$ blames\change{\footnote{Here $\lvert H\rvert$ denotes the honest \emph{stake}, which equals $S_f+1$ since $S=2S_f+1$; under the uniform-stake convention of \Cref{sec:gossip} (each validator holds $S/n$ stake) this is a strict majority of the guard stake, i.e., the honest-guard count $n-f=f+1$ when $n=2f+1$.}} for all $v\in\cap_{g\in H} \asleep_g(r)$, where $\lvert\cap_{g\in H} \asleep_g(r)\rvert \ge f_c+1$. \change{Since $S_f < S/2$, attesting stake of $S_f+1$ constitutes a strict majority of total guard stake, satisfying the $\Call{Maj}{\cdot}$ condition in \textsc{OnLBlame}; each such $v$ is therefore added to $\lblamed(r)$.}
\end{proof}

\begin{proof}[Proof of \Cref{thm:recovery}]
    We use a black-box synchronous BA primitive over the guard set that under the $n=2f+1$ honest-majority assumption provides \emph{agreement} (all honest guards output the same value), \emph{validity} (if all honest guards input the same valid blameset, that value is output), and \emph{termination} (within $\Delta_{BA}$).

    \noindent\emph{Part (i).} By \Cref{thm:guard-accountable-safety} (resp.\ \Cref{thm:guard-accountable-liveness}), every honest guard holds at least one valid blameset of size $\ge f_c+1$ within the stated time bound. \change{Each guard $i$ executes \textsc{Recover}, which broadcasts $S_i$ via $\Call{BC}{i, S_i, \cdot}$ and then collects $\mathsf{RCVec}_i = \langle S_k\rangle_{k\in[n]}$: entry $S_k$ is the unique blameset received from guard $k$'s BC instance within $\Delta$, or $\bot$ if no such message arrived. By the Agreement property of synchronous BC, all honest guards receive the same value from each guard $k$'s broadcast, so every honest guard assembles the same vector $\mathsf{RCVec}$. For each $k\in[n]$, all honest guards run the same BA instance $\mathsf{BA}(k, S_k)$ on identical input; by BA Agreement, all honest guards output the same $SR_k$. Since every honest guard applies the deterministic $\Call{Order}{\{SR_k\}_{k\in[n]}}$ to the same set of outputs, they obtain the same ordering $\mathbf{O}$ and select the same first valid entry as $\mathcal{B}$. At least one guard $k^*$ holds a valid blameset ($S_{k^*}\ne\bot$), so at least one $SR_{k^*}$ is valid in $\mathbf{O}$, ensuring $\mathcal{B}$ is well-defined.} By Lemmas~\ref{lem:guard-honest-safe} and~\ref{lem:guard-honest-live}, no honest validator belongs to any valid blameset, so $\mathcal{B}$ contains only faulty validators.

    \noindent\emph{Part (ii).} Since the up-to-$3f_c$ corruptions include the $\ge f_c+1$ removed validators, at most $3f_c-(f_c+1)=2f_c-1$ corruptions remain among the $4f_c$ survivors, leaving $\ge 2f_c+1$ honest (majority).

    \noindent\emph{Part (iii).} By the fork-choice rule and \Cref{thm:resilient-safety}, recovery adopts the unique CheckpointQC chain (\Cref{lem:unique-checkpoint}); since every FinalityQC-finalized height is backed by a CheckpointQC whose $4f_c+1$ signatories include $\ge f_c+1$ honest validators (with at most $3f_c$ corrupt), and whose existence is unique by \Cref{lem:unique-checkpoint}, it lies on the adopted chain and is preserved.
\end{proof}

\subsection{Recovery Protocol Details}\label{sec:guard-recovery-details}

Once all honest guards agree on blameset $\mathcal{B}$ (\Cref{thm:guard-accountable-safety,thm:guard-accountable-liveness}), the three recovery steps proceed as follows.

\emph{Step R1 (Fork-choice rule).}
In the case of a safety violation, the canonical chain is determined by the unique CheckpointQC of the highest height for which a FinalityQC exists (\Cref{thm:resilient-safety}).
All core validators not in $\mathcal{B}$ adopt this checkpoint as the new starting state.
If no FinalityQC exists yet (the fork happened before any checkpoint was finalized), validators adopt the most recent CheckpointQC whose underlying committed leader slot is agreed upon by at least $2f_c+1$ honest validators; this is well-defined since \Cref{lem:unique-checkpoint} guarantees at most one CheckpointQC per height.
In the case of a pure liveness violation (no safety fork), the chain continues from the last finalized checkpoint without rollback.

\emph{Step R2 (Reconfiguration).}
Guards broadcast the agreed blameset via Byzantine Broadcast to all validators.
Every core validator not in $\mathcal{B}$ records the new validator set $\mathcal{V}' = \mathcal{V} \setminus \mathcal{B}$, where $|\mathcal{V}'| \le 4f_c$.
Because the blameset contains only provably misbehaving validators (Lemmas~\ref{lem:guard-honest-safe},~\ref{lem:guard-honest-live}) and removes at least $f_c+1$ of the up-to-$3f_c$ corrupt validators, the remaining set has at most $2f_c-1$ malicious and at least $2f_c+1$ honest validators out of $4f_c$.
The reconfigured core thus enjoys a strict honest majority ($2f_c+1 > 4f_c/2$), allowing a synchronous honest-majority ($n'=2f'+1$) BFT protocol, such as Sync HotStuff~\cite{abraham2020synchotstuff}, which tolerates $f' < n'/2$ faults, to safely resume consensus.

\emph{Step R3 (Client guarantees post-recovery).}
Checkpoint-path clients are unaffected by the fork: by \Cref{thm:resilient-safety}, the CheckpointQC chain is unique and all state finalized by a FinalityQC is preserved.
Fast-path clients may need to roll back transactions that were committed after the last finalized checkpoint but before the fork was resolved; however, the number of such transactions is bounded by the volume of transactions processed between two synchronous timestamping events, consistent with the economic security model of CoBRA~\cite{cobra}.

\para{Comparison with prior recovery work}
Our recovery mechanism is conceptually related to synchronous recovery gadgets studied in~\cite{civit2025recover, gong2025recovery} and the accountability-based recovery of~\cite{sheng2023bft}. Unlike~\cite{civit2025recover}, which assumes a fixed known corruption bound throughout recovery, \gossip's blameset construction actively reduces the effective corruption count, allowing subsequent consensus to run under a stronger honest-majority assumption. Compared to~\cite{gong2025recovery}, our approach ties recovery directly to the FinalityQC chain, avoiding ambiguity in the fork-choice rule under concurrent recoveries. We consider a deeper formal comparison with these works an important direction for future work.

\section{Implementation and Evaluation Details}

\subsection{Implementation} \label{sec:implementation}

We implement a networked, multi-core \lmysticeti validator in Rust by forking the Mysticeti codebase~\cite{mysticeti-code}. Our implementation leverages \texttt{tokio}~\cite{tokio} for asynchronous networking, utilizing raw TCP sockets for communication without relying on any RPC frameworks. For cryptographic operations, we rely on \texttt{ed25519-consensus}~\cite{ed25519-consensus} for asymmetric cryptography and \texttt{blake2}~\cite{rustcrypto-hashes} for cryptographic hashing. To ensure data persistence and crash recovery, we employed the Write-Ahead Log (WAL). Our WAL optimizes I/O operations through vectored writes~\cite{writev} and efficient memory-mapped file usage with the \texttt{minibytes}~\cite{minibytes} crate, minimizing data copying and serialization.

In addition to regular unit tests, we inherited and utilized two supplementary testing utilities from the Mysticeti codebase. First, a simulation layer replicates the functionality of the \texttt{tokio} runtime and TCP networking. This simulated network accurately simulates real-world WAN latencies, while the \texttt{tokio} runtime simulator employs a discrete event simulation approach to mimic the passage of time. Second, a command-line utility (called \emph{orchestrator}) which deploys real-world clusters of \lmysticeti on machines distributed across the globe.

We are open-sourcing our \lmysticeti implementation, along with its simulator and orchestration tools, to ensure reproducibility of our results\footnote{\codelink}.

\subsection{Experimental Setup} \label{sec:experimental-setup}
This section complements \Cref{sec:evaluation} by detailing the experimental setup used to evaluate \lmysticeti.

We deploy \lmysticeti on AWS, using \texttt{m5d.8xlarge} instances across $13$ different AWS regions: Northern Virginia (us-east-1), Oregon (us-west-2), Canada (ca-central-1), Frankfurt (eu-central-1), Ireland (eu-west-1), London (eu-west-2), Paris (eu-west-3), Stockholm (eu-north-1), Mumbai (ap-south-1), Singapore (ap-southeast-1), Sydney (ap-southeast-2), Tokyo (ap-northeast-1), and Seoul (ap-northeast-2). Validators are distributed across those regions as equally as possible.
Each machine provides $10$\,Gbps of bandwidth, $32$ virtual CPUs (16 physical cores) on a $3.1$\,GHz Intel Xeon Skylake 8175M, $128$\,GB memory, and runs Linux Ubuntu server $24.04$.
We select these machines because they provide decent performance, are in the price range of ``commodity servers'', and match the minimal specifications of modern quorum-based blockchains~\cite{sui-min-specs}.

In \Cref{sec:evaluation}, \emph{latency} refers to the time elapsed from the moment a client submits a transaction to when it is committed by the validators, and \emph{throughput} refers to the number of transactions committed per second.
We instantiate several geo-distributed benchmark clients within each validator submitting transactions in an open loop model, at a fixed rate. We experimentally increase the load of transactions sent to the systems, and record the throughput and latency of commits. As a result, all plots \Cref{sec:evaluation} illustrate the steady-state latency of all systems under low load, as well as the maximal throughput they can provide after which latency grows quickly. Transactions in the benchmarks are arbitrary and contain $512$ bytes. We configure both \lmysticeti and Mysticeti with $2$ leaders per round, and a leader timeout of 1 second.
\section{Asynchronous \lmysticeti} \label{sec:async}

We present \asynclmysticeti, the asynchronous variant of \lmysticeti. It builds upon the same ideas as \lmysticeti, but operates in a fully asynchronous network model and leverages a threshold common coin to achieve liveness. \asynclmysticeti is the first $5f+1$-validator BFT consensus protocol to achieve both safety and liveness in a fully asynchronous network. It achieves low latency through its shortened commit path. \asynclmysticeti algorithm is \Cref{alg:main}, \Cref{alg:decider}, and \Cref{alg:helper} when we also include the orange-colored lines.

\subsection{Additional Assumptions} \label{sec:model}

First, the communication network is asynchronous and messages can be delayed arbitrarily, but messages among honest validators are eventually delivered. 

Additionally, we employ a global perfect coin to introduce randomization, similar to previous work~\cite{blum2020asynchronous,cachin2000random,dag-rider,loss2018combining}. This coin can be constructed using an adaptively secure threshold signature scheme~\cite{bacho2022on,boneh2001short}, with the distributed key setup performed under asynchronous conditions~\cite{abraham2023bingo,abraham2023reaching}.

As with \lmysticeti (proven in \Cref{sec:proofs} of the appendix), \asynclmysticeti implements BAB~\cite{cristian1995}.

All notation in the proofs below refers to the core layer; for clarity, we drop the subscript and write $f$ and $n$ in place of $f_c$ and $n_c$.




\subsection{Safety Lemmas}
We start by proving the Total Order and Integrity properties of BAB. A crucial intermediate result towards these properties is that all honest validators have consistent commit sequences, i.e., the committed sequence of one honest validator is a prefix of another's, or vice-versa. This is shown in~\Cref{lem:agree-commit} and~\Cref{lem:consistent}, which the following lemmas and observations build up to.

\begin{definition}[Weak Certificate]\label{def:weak-cert}
    A \emph{weak certificate} for a block $b$ in round $r$ is a set of $2f+1$ distinct \emph{valid} round-$(r+2)$ blocks, each from a \emph{distinct validator}, that support $b$.
\end{definition}

\begin{definition}[Strong Certificate]\label{def:strong-cert}
    A \emph{strong certificate} for a block $b$ in round $r$ is a set of $4f+1$ distinct \emph{valid} round-$(r+2)$ blocks, each from a \emph{distinct validator}, that support $b$.
\end{definition}

\begin{lemma}[Strong-to-Weak Propagation] \label{lma1_cert_path}
    Let $b$ be a round-$r$ block with a \emph{strong certificate} $S$ (i.e., $|S|=4f+1$ distinct round-$(r+2)$ blocks supporting $b$, \Cref{def:strong-cert}). Then every valid block at any future round $r'\ge r+3$ has paths to at least $2f+1$ blocks from $S$; equivalently, it indirectly references a \emph{weak certificate} for $b$ (\Cref{def:weak-cert}).
\end{lemma}
\iflong
    \begin{proof}
        We prove the claim by induction on $r'$, starting at $r'=r+3$.

        Base case ($r'=r+3$): Let $x$ be a valid round-$(r+3)$ block. By construction, $x$ references $4f+1$ distinct valid round-$(r+2)$ blocks from distinct validators; thus at most $f$ of these can be Byzantine, so $x$ includes at least $3f+1$ round-$(r+2)$ blocks from honest validators. Likewise, the strong certificate $S$ contains at least $3f+1$ blocks from honest validators (since at most $f$ of its $4f+1$ members can be Byzantine). Restricting attention to honest validator identities (a universe of size $4f+1$), these two sets intersect in at least $(3f+1)+(3f+1)-(4f+1)=2f+1$ blocks. Hence $x$ directly references at least $2f+1$ members of $S$, so $x$ has paths to a weak certificate for $b$.

        Induction step: Assume that every valid round-$r'$ block, for some $r'\ge r+3$, has paths to at least $2f+1$ members of $S$. Consider any valid round-$(r'+1)$ block $y$. By construction, $y$ references $4f+1$ distinct round-$r'$ blocks. By the induction hypothesis, each of those round-$r'$ blocks has paths to at least $2f+1$ members of $S$. Therefore, $y$ references at least one round-$r'$ block that has paths to at least $2f+1$ members of $S$; by transitivity of paths, $y$ itself has paths to at least $2f+1$ members of $S$, i.e., to a weak certificate for $b$.
    \end{proof}
\fi
\begin{observation}\label{obs:no-equiv-vote}
    A block cannot support for more than one block proposal from a given validator, in a given round.
\end{observation}
\iflong
    \begin{proof}
        This is by construction. Honest validators interpret support in the DAG through deterministic depth-first traversal. So even if a block $b$ in the vote round has paths to multiple leader round blocks from the same validator $v$ (i.e., equivocating blocks), all honest validators will interpret $b$ to vote for only one of $v$'s blocks (the first block to appear in the depth-first traversal starting from $b$).
    \end{proof}
\fi

\begin{lemma}[Strong certificate exclusivity] \label{lem:one-cert-per-validator}
    Fix a validator $v$ and round $r$. If some round-$r$ block $b$ of $v$ has a strong certificate, then no other round-$r$ block $b'\ne b$ of $v$ can have even a weak certificate. In particular, at most one round-$r$ block of $v$ can have a strong certificate.
\end{lemma}
\iflong
    \begin{proof}
        Suppose, for contradiction, that $b$ has a strong certificate $S$ (so $|S|=4f+1$ decision-round blocks at round $r+2$ support $b$), and some other block $b'\ne b$ from the same validator and round has a weak certificate $W$ (so $|W|=2f+1$ vote-round blocks at round $r+2$ support $b'$). Both $S$ and $W$ are over the universe of $5f+1$ validator identities and contain one valid block per identity.

        Let $\mathrm{Id}(S)$ and $\mathrm{Id}(W)$ be the sets of validator identities appearing in the certificates $S$ and $W$ (one valid round-$(r+2)$ block per identity). Then $|\mathrm{Id}(S)|=4f+1$ and $|\mathrm{Id}(W)|=2f+1$. By quorum intersection over identities,
        $$
            |\mathrm{Id}(S) \cap \mathrm{Id}(W)| \ge (4f+1) + (2f+1) - (5f+1) = f+1 > 0.
        $$
        Since at most $f$ validators are Byzantine overall, this intersection contains at least one honest identity. Pick such an honest validator $u\in \mathrm{Id}(S)\cap \mathrm{Id}(W)$, and let $x$ be $u$'s unique valid decision-round (round $r+2$) block (by the uniqueness-per-round assumption). Because certificates record one valid block per identity, the same block $x$ is the representative of $u$ in both $S$ and $W$, hence $x$ supports both $b$ and $b'$. This contradicts \Cref{obs:no-equiv-vote}. Therefore, $b'$ cannot have a weak certificate. Taking $W$ to be a strong certificate $S'$ gives the special case that two distinct blocks of $v$ in the same round cannot both have strong certificates.
    \end{proof}
\fi

\begin{observation}\label{obs:exists-certificate}
    If an honest validator $v$ directly or indirectly commits a block $b$, then $v$'s local DAG contains a weak certificate for $b$.
\end{observation}
\iflong
    \begin{proof}
        This follows immediately from our direct and indirect commit rules.
    \end{proof}
\fi

\begin{observation}\label{obs:agree-order}
    Honest validators agree on the sequence of leader slots.
\end{observation}
\iflong

    \begin{proof}
        This follows immediately from the properties of the common coin, see~\Cref{sec:model}.
    \end{proof}
\fi

\begin{observation}[Unique valid block per (validator, round)]\label{obs:unique-valid}
    For any validator $w$ and round $r$, at most one block signed by $w$ is counted as valid in round $r$. Since commit rules apply only to valid blocks, any commitment for a slot with leader $w$ in round $r$ must be that unique valid block.
\end{observation}

\iflong
    \begin{proof}
        Immediate from the DAG construction and the fact that commit rules consider only valid blocks.
    \end{proof}
\fi

\begin{lemma} \label{lma4_honest_unique}
    If an honest validator $v$ commits some block $b$ in a slot $s$, then no other honest validator decides to directly skip the slot $s$.
\end{lemma}
\iflong
    \begin{proof}
        Assume by contradiction that some honest validator $v'$ decides to directly skip $s$. By the direct-skip rule, this means that in $v'$'s local DAG there exists a set $N$ of $4f+1$ distinct valid round-$(r+2)$ blocks (one per validator identity) whose support does not go to $b$ (a direct-skip witness for $s$).

        Since $v$ commits $b$ at $s$, by \Cref{obs:exists-certificate} there exists a weak certificate $W$ for $b$ at $s$ in $v$'s local DAG: a set of $2f+1$ distinct valid round-$(r+2)$ blocks (one per validator identity) that support $b$ (\Cref{def:weak-cert}). Consider the identity sets $\mathrm{Id}(N)$ and $\mathrm{Id}(W)$. We have $|\mathrm{Id}(N)|=4f+1$ and $|\mathrm{Id}(W)|=2f+1$, hence by quorum intersection over the $5f+1$ validator identities,
        $$
            |\mathrm{Id}(N) \cap \mathrm{Id}(W)| \ge (4f+1) + (2f+1) - (5f+1) = f+1.
        $$
        As at most $f$ validators are Byzantine, the intersection contains at least one honest identity. Pick such an honest validator $u\in \mathrm{Id}(N)\cap\mathrm{Id}(W)$, and let $x$ be $u$'s unique valid round-$(r+2)$ block.

        Because certificates/witnesses record one valid block per identity, the same block $x$ is the representative of $u$ in both $N$ and $W$. By definition of $W$, $x$ supports $b$; by definition of $N$, the very same $x$ is counted as not supporting $b$ at $v'$. By \Cref{obs:no-equiv-vote} (uniqueness of support per validator/round under the deterministic rule), $x$ cannot both support and not support $b$. Hence $v'$ cannot have a valid direct-skip witness $N$, a contradiction.
    \end{proof}
\fi

\begin{lemma} \label{lma5_honest_skip}
    If an honest validator directly commits some block in a slot $s$, then no other honest validator decides to skip the slot $s$.
\end{lemma}
\iflong
    \begin{proof}
        Assume by contradiction that an honest validator $v$ directly commits block $b$ in slot $s$ while another honest validator $v'$ decides to skip $s$. By \Cref{lma4_honest_unique}, $v'$ cannot directly skip $s$; therefore $v'$ must attempt to skip $s$ via the indirect decision rule. Let $r$ be the round of $s$.

        Since $v$ directly commits $b$, there exists a strong certificate $S$ for $b$ at $s$ (i.e., $|S|=4f+1$ distinct valid round-$(r+2)$ blocks supporting $b$). By \Cref{lma1_cert_path} (Strong-to-Weak Propagation), every valid block at any round $r'\ge r+3$ has paths to at least $2f+1$ members of $S$, i.e., it carries a weak certificate for $b$. In particular, any valid anchor block for deciding $s$ (which necessarily lies at some round $r'\ge r+3$) has paths to a weak certificate for $b$.

        But the indirect skip rule for $s$ requires an anchor whose deterministic support excludes $b$ (equivalently, with no weak certificate for $b$). This contradicts the propagation property above. Hence $v'$ cannot skip $s$ indirectly either. Contradiction.
    \end{proof}
\fi

\begin{lemma}\label{lem:agree-commit}
    If a slot $s$ is committed at two honest validators, then $s$ contains the same block at both validators.
\end{lemma}
\iflong
    \begin{proof}
        Let $v$ and $u$ be honest validators and suppose $v$ commits block $b$ at slot $s$. If $u$ commits $s$ as well, we show that $u$ commits $b$.

        Let $w$ be the validator identity of the leader for $s$ (i.e., the creator of $b$). By \Cref{obs:agree-order}, all honest validators agree on the leader identity per slot, so $u$ also agrees that $s$ must contain a block by $w$.

        By \Cref{obs:unique-valid}, there is a unique valid round-$r$ block by $w$ that can be committed in $s$.

        Since $v$ commits $b$ and $b$ is by $w$, $b$ is this unique valid block by $w$ for round $r$. Hence if $u$ commits $s$, it must also commit $b$. Thus $s$ contains the same block at both validators.
    \end{proof}
\fi

We say that a slot is \textit{decided} at a validator $v$ if $s$ is committed or skipped, that is, if it is categorized as $\scommit$ or $\sskip$. Otherwise, $s$ is \textit{undecided}.

\begin{lemma}\label{lem:consistent}
    If a slot $s$ is decided at two honest validators $v$ and $v'$, then either both validators commit $s$, or both validators skip $s$.
\end{lemma}
\iflong
    \begin{proof}
        Assume by contradiction that there exists a slot $s$ such that $v$ and $v'$ decide differently at $s$. We consider a finite execution prefix and assume \textit{wlog} that $s$ is the highest slot at which $v$ and $v'$ decide differently (*). Further assume \textit{wlog} that $v$ commits $s$ and $v'$ skips $s$. By \Cref{lma4_honest_unique} and \Cref{lma5_honest_skip}, neither $v$ nor $v'$ could have used the direct decision rule for $s$; they must both have used the indirect rule. Consider now the anchor of $s$: $v$ and $v'$ must agree on which slot is the anchor of $s$, since by our assumption (*) above, they make the same decisions for all slots higher than $s$, including the anchor of $s$. Let $s'$ be the anchor of $s$; $s'$ must be committed at both $v$ and $v'$. Thus, by \Cref{lem:agree-commit}, $v$ and $v'$ commit the same block $b'$ at $s'$. But then $v$ and $v'$ cannot reach different decisions about slot $s$ using the indirect decision rule. We have reached a contradiction.
    \end{proof}
\fi

We have proven the consistency of honest validators' commit sequences: honest validators commit (or skip) the same leader blocks, in the same order. However, we are not done: we also need to prove that non-leader blocks are delivered in the same order by honest validators. We show this next.

\textbf{Causal history and delivery conditions}
Consider an honest validator $v$. We call the \textit{causal history} of a block $b$ in $v$'s DAG, the transitive closure of all blocks referenced by $b$ in $v$'s DAG, including $b$ itself. In \asynclmysticeti, a block $b$ is delivered by an honest validator $v$ if (1) there exists a committed leader block $l$ in $v$'s DAG such that $b$ is in $l$'s causal history (2) all slots up to $l$ are decided in $v$'s DAG and (3) $b$ has not been delivered as part of a lower slot's causal history. In this case we say $b$ is \textit{delivered at} slot $s$, or \textit{delivered with} block $l$.

\begin{lemma}\label{lem:delivered-same-slot}
    If a block $b$ is delivered by two honest validators $v$ and $v'$, then $b$ is delivered at the same slot $s$, and $b$ is delivered with the same leader block $l$, at both $v$ and $v'$.
\end{lemma}
\iflong
    \begin{proof}
        Let $s$ be the slot at which $b$ is delivered at validator $v$, and $l$ the corresponding leader block in $s$, also at validator $v$. Consider now the slot $s'$ at which $b$ is delivered at validator $v'$, and $l'$ the corresponding leader block. Assume by contradiction that $s' \ne s$. If $s' < s$, then $v$ would have also delivered $b$ at slot $s'$, since by \Cref{lem:agree-commit} must commit the same leader blocks in the same slots, so $v$ could not have delivered $b$ again at slot $s$; a contradiction. Similarly, if $s < s'$, then $v'$ would have already delivered $b$ at slot $s$, since by \Cref{lem:agree-commit} $v$ and $v'$ must have committed the same block in slot $s$; contradiction. Thus it must be that $s = s'$, and by \Cref{lem:agree-commit}, $l = l'$.
    \end{proof}
\fi

We can now prove the main safety properties of \sysname: Total Order and Integrity.
\begin{theorem}[Total Order]
    \asynclmysticeti satisfies the total order property of Byzantine Atomic Broadcast.
\end{theorem}
\iflong
    \begin{proof}
        This property follows immediately from \Cref{lem:delivered-same-slot} and the fact that honest validators order the causal histories of committed blocks using the same deterministic function, and deliver blocks in this order.
    \end{proof}
\fi

\begin{theorem}[Integrity]\label{thm:integrity}
    \asynclmysticeti satisfies the integrity property of Byzantine Atomic Broadcast.
\end{theorem}
\iflong
    \begin{proof}
        This is by construction: a block $b$ is delivered as part of the causal history of a committed leader block only if $b$ has not been delivered along with an earlier leader block (see "Causal history \& delivery conditions" above). So an honest validator cannot deliver the same block twice.
    \end{proof}
\fi

\subsection{Liveness Lemmas}

\para{Block inclusion.}
The following two lemmas establish that blocks broadcast by honest validators are eventually included in all honest validators' DAGs.

\begin{lemma} \label{lem:causal-history-inclusion}
    If a block $b$ produced by an honest validator $v$ references some block $b'$, then $b'$ will eventually be included in the local DAG of every honest validator.
\end{lemma}
\iflong

    \begin{proof}
        This is ensured by the synchronizer sub-component in each validator: if some validator $w$ receives $b$ from $v$, but does not have $b'$ yet, $w$ will request $b'$ from $v$; since $v$ is honest and the network links are reliable, $v$ will eventually receive $w$'s request, send $b'$ to $w$, and $w$ will eventually receive $b'$. The same is recursively true for any blocks from the causal history of $b'$, so $w$ will eventually receive all blocks from the causal history of $b'$ and thus include $b'$ in its local DAG.
    \end{proof}
\fi

\begin{lemma} \label{lem:block-inclusion}
    If a honest validator $v$ broadcasts a block $b$, then $b$ will eventually be included in the local DAG of every honest validator.
\end{lemma}
\iflong
    \begin{proof}
        Since network links are reliable, all honest validators will eventually receive $b$ from $v$. By \Cref{lem:causal-history-inclusion}, all honest validators will eventually receive all of $b$'s causal history, and so will include $b$ in their local DAG.
    \end{proof}
\fi

\paragraph{Main structural lemmas.}
These are the main structural lemmas that we will use to prove liveness. The key idea is that the reference rule creates significant overlap among honest blocks across consecutive rounds, which we can leverage to ensure that a randomly chosen leader has sufficient honest support to be directly committed.

\begin{lemma}[Reference Honesty Lower Bound]\label{lem:valid-ref-honest-lb}
    Any \emph{valid} round-$(R+1)$ block (whether created by an honest or Byzantine validator) references at least $3f+1$ honest round-$R$ blocks.
\end{lemma}
\iflong
    \begin{proof}
        In any round there are $5f+1$ total validators of which at most $f$ are Byzantine, so at least $4f+1$ are honest. A valid round-$(R+1)$ block, by definition of the protocol's formation rule, references at least $4f+1$ distinct round-$R$ blocks (one per validator identity). At most $f$ of these can be Byzantine, hence at least $4f+1-f = 3f+1$ are honest.
    \end{proof}
\fi

\begin{lemma}[Many Heavily Referenced Round-$R$ Blocks]\label{lem:heavy-honest-R-blocks}
    There are at least $2f+1$ honest round-$R$ blocks each referenced by at least $f+1$ honest round-$R+1$ blocks.
\end{lemma}
\iflong
    \begin{proof}
        By \Cref{lem:valid-ref-honest-lb} every valid round-$(R+1)$ block references at least $3f+1$ \emph{honest} round-$R$ blocks. There are $4f+1$ honest validators, hence $4f+1$ honest round-$(R+1)$ blocks in total. Consider the bipartite graph whose left vertices $A$ are the honest round-$(R+1)$ blocks and whose right vertices $B$ are the honest round-$R$ blocks; connect $a\in A$ to $b\in B$ if $a$ references $b$.

        Let $E$ be the total number of edges. Each $a\in A$ has degree at least $3f+1$, so
        $$
            E \ge (4f+1)(3f+1).
        $$
        Suppose for contradiction that fewer than $2f+1$ honest round-$R$ blocks have degree at least $f+1$ (i.e., are referenced by $\ge f+1$ honest round-$(R+1)$ blocks). Let $X \subseteq B$ be the (assumed) set of degree $\ge f+1$ blocks with $|X| \le 2f$. Any block in $B\setminus X$ then has degree at most $f$.

        We upper bound $E$ under this hypothesis:
        \begin{align*}
            E & \le |X|(4f+1) + (|B|-|X|)f \\
              & = |X|(4f+1 - f) + f(4f+1)  \\
              & = |X|(3f+1) + f(4f+1).
        \end{align*}
        Using $|X| \le 2f$ we get
        \begin{align*}
            E & \le 2f(3f+1) + f(4f+1)     \\
              & = (6f^2 + 2f) + (4f^2 + f) \\
              & = 10f^2 + 3f.
        \end{align*}
        Yet
        $$
            (4f+1)(3f+1) = 12f^2 + 7f + 1 > 10f^2 + 3f,
        $$
        a contradiction. Therefore, $|X| \ge 2f+1$, establishing the claim.
    \end{proof}
\fi

\begin{lemma}[Common Honest Ancestors]\label{lem:common-honest-ancestors}
    Any set of $4f+1$ round-$R+2$ blocks collectively references at least $2f+1$ common honest round-$R$ blocks.
\end{lemma}
\iflong
    \begin{proof}
        Let $H$ denote the set of (at least) $2f+1$ honest round-$R$ blocks guaranteed by \Cref{lem:heavy-honest-R-blocks}: every $h\in H$ is referenced by at least $f+1$ honest round-$(R+1)$ blocks.

        Consider any multiset $S$ of $4f+1$ round-$(R+2)$ blocks (they may be arbitrary, honest or Byzantine). By applying \Cref{lem:valid-ref-honest-lb} to round $R+1$ vs. $R+2$, each valid round-$(R+2)$ block includes at least $3f+1$ \emph{honest} round-$(R+1)$ blocks among the $4f+1$ distinct references it must carry.

        Fix some $h\in H$. Suppose, for contradiction, that a particular $s\in S$ fails to (indirectly) reference $h$. Then $s$ must omit every honest round-$(R+1)$ block that references $h$. However, $h$ has at least $f+1$ such honest round-$(R+1)$ children, while $s$ can omit at most $(4f+1) - (3f+1) = f$ honest round-$(R+1)$ blocks (because it necessarily includes at least $3f+1$ of the $4f+1$ honest ones). Since $f+1 > f$, omitting them all is impossible. Therefore, every $s\in S$ (indirectly) references $h$.

        The argument holds for each $h\in H$, so all blocks in $S$ commonly (indirectly) reference every element of $H$. Thus, their intersection over round-$R$ ancestors contains $H$, and has size at least $|H| \ge 2f+1$.
    \end{proof}
\fi
\paragraph{Liveness theorems.}
We now leverage the structural overlap to obtain probabilistic liveness via the common coin.

\begin{definition}[Core Set of Round $R$]\label{def:core-set}
    Let $C_R$ be the set of honest round-$R$ blocks each referenced by at least $f+1$ honest round-$(R+1)$ blocks. By \Cref{lem:heavy-honest-R-blocks}, $|C_R| \ge 2f+1$.
\end{definition}

\begin{lemma}[Persistence of Core Support]\label{lem:core-persistence}
    Every valid round-$(R+2)$ block (indirectly) references every block in $C_R$.
\end{lemma}
\iflong
    \begin{proof}
        Immediate from \Cref{lem:common-honest-ancestors} since $C_R \subseteq H$ for the set $H$ used there.
    \end{proof}
\fi

We denote by $l \leq 5f + 1$ the number of leader slots per round.

\begin{lemma}[Direct Commitment via Core Intersection]\label{lem:direct-core-intersection}
    Fix a round $r$ and let $n = 5f+1$. When $l$ leader slots are sampled uniformly at random without replacement from the $n$ validators, the probability that at least one slot can be directly committed is
    $$
        p^\star = 1 - \frac{\binom{n - |C_R|}{l}}{\binom{n}{l}}.
    $$
    Moreover, if $l > n - |C_R|$ (in particular, if $l > 3f$ using $|C_R| \ge 2f+1$) then $p^\star = 1$ (deterministic success).
\end{lemma}
\iflong
    \begin{proof}
        By \Cref{lem:core-persistence}, any selected leader whose block lies in $C_R$ is (eventually) directly commit-able by every honest validator. Thus a successful direct commitment in round $r$ occurs iff the sampled set intersects $C_R$. The probability that it does not intersect $C_R$ is exactly the hypergeometric zero-success probability $\binom{n - |C_R|}{l}/\binom{n}{l}$; subtracting from $1$ yields $p^\star$.

        If $l > n - |C_R|$, it is impossible to choose all $l$ leaders outside $C_R$, so the intersection is certain and $p^\star = 1$. Using only $|C_R| \ge 2f+1$, we have $n - |C_R| \le 5f+1 - (2f+1) = 3f$, hence any $l > 3f$ suffices for determinism.
    \end{proof}
\fi

\begin{lemma}[Eventual Slot Resolution]\label{lem:eventual-slot-resolution}
    Fix a slot $s$. Every honest validator eventually either commits or skips $s$, with probability 1.
\end{lemma}
\iflong
    \begin{proof}
        We prove the lemma by showing that the probability of $s$ remaining undecided forever at some honest validator is 0. In order for $s$ to remain undecided forever, $s$ cannot be committed or skipped directly. Furthermore, $s$ cannot be decided using the indirect rule. This means that the anchor $s'$ of $s$ must also remain undecided forever, and therefore the anchor $s''$ of $s'$ must remain undecided forever, and so on. The probability of this occurring is at most equal to the probability of an infinite sequence of rounds with no directly committed slots, equal to $\lim_{t \to \infty}(1 - p^\star)^t = 0$, where $p^\star > 0$ is the probability from \Cref{lem:direct-core-intersection}.
    \end{proof}
\fi

\begin{theorem}[Validity]
    \asynclmysticeti satisfies the validity property of Byzantine Atomic Broadcast.
\end{theorem}
\iflong
    \begin{proof}
        Let $v$ be an honest validator and $b$ a block broadcast by $v$. We show that, with probability 1, $b$ is eventually delivered by every honest validator. \Cref{lem:block-inclusion} $b$ is eventually included in the local DAG of every honest validator. So every honest validator will eventually include a reference to $b$ in at least one of its blocks. Let $r$ be the highest round at which some honest validator includes a reference to $b$ in one of its blocks. By \Cref{lem:direct-core-intersection}, with probability 1, eventually some block $b'$ at a round $r' > r$ will be directly committed. Block $b'$ must reference at least $4f + 1$ blocks, thus at least $3f + 1$ blocks from honest validators. Since all validators have $b$ in their causal histories by round $r$, $b'$ must therefore have a path to $b$. \Cref{lem:eventual-slot-resolution} guarantees that all slots before $b'$ are eventually decided, so $b'$ is eventually delivered. Thus, $b$ will be delivered at all honest validators at the latest when $b'$ is delivered along with its causal history.
    \end{proof}
\fi

\begin{theorem}[Agreement]
    \asynclmysticeti satisfies the agreement property of Byzantine Atomic Broadcast.
\end{theorem}
\iflong
    \begin{proof}
        Let $v$ be an honest validator and $b$ a block delivered by $v$. We show that, with probability 1, $b$ is eventually delivered by every honest validator. Let $l$ be the leader block with which $b$ is delivered, and $s$ the corresponding slot. By \Cref{lem:eventual-slot-resolution}, all blocks up to and including $s$ are eventually decided by all honest validators, with probability 1. By \Cref{lem:delivered-same-slot}, all honest validators commit $l$ in $s$. Therefore, all honest validators deliver $b$ eventually.
    \end{proof}
\fi

\subsection{\asynclmysticeti Performance}

We implement\footnote{\asynccodelink} \asynclmysticeti by extending our \lmysticeti prototype (Appendix~\ref{sec:implementation}) and evaluate it using the same setup as in \Cref{sec:evaluation}.

\Cref{fig:evaluation-async} presents the WAN throughput-latency results for \lmysticeti and \asynclmysticeti, with Mysticeti and Hydrangea included as visual reference points from \Cref{sec:evaluation}. We consider two scenarios: (a) committees of $10$ and $50$ validators without faults (\Cref{fig:evaluation-common-async}), and (b) committees of $10$ and $11$ validators (enough to tolerate $3$ and $2$ faults) under corresponding crash failures (\Cref{fig:evaluation-faults-async}). We cap throughput at $50{,}000$\,tps for cost reasons.
\asynclmysticeti matches Mysticeti's throughput--latency performance across both scenarios and committee sizes, as both commit in three rounds in the common case. This highlights a key trade-off: for equivalent performance, one can choose between an asynchronous protocol with $n=5f+1$ or a partially synchronous one with $n=3f+1$.

\begin{figure}[t]
    \centering
    \begin{subfigure}{0.49\textwidth}
        \centering
        \includegraphics[width=\linewidth]{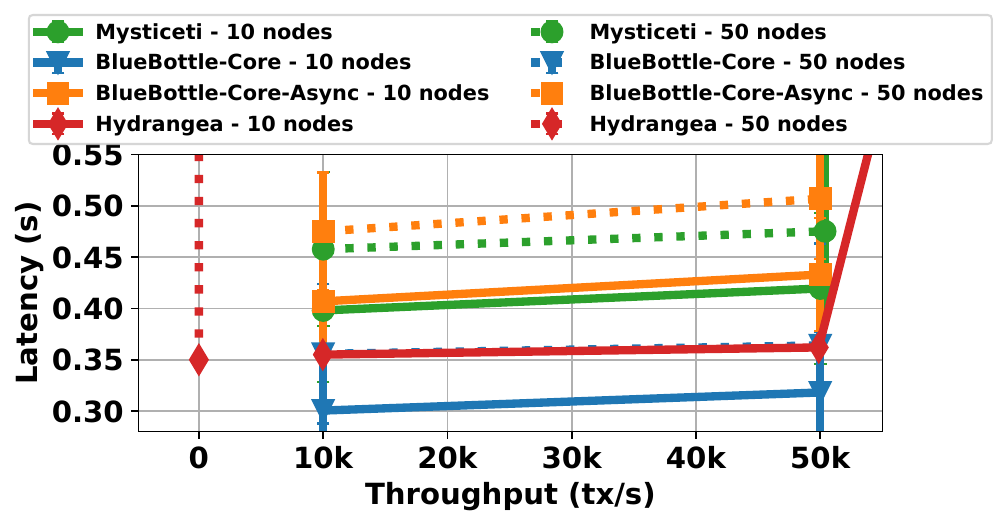}
        \caption{Committees with $10$ and $50$ validators with no validator faults.}
        \label{fig:evaluation-common-async}
    \end{subfigure}\hfill
    \begin{subfigure}{0.49\textwidth}
        \centering
        \includegraphics[width=\linewidth]{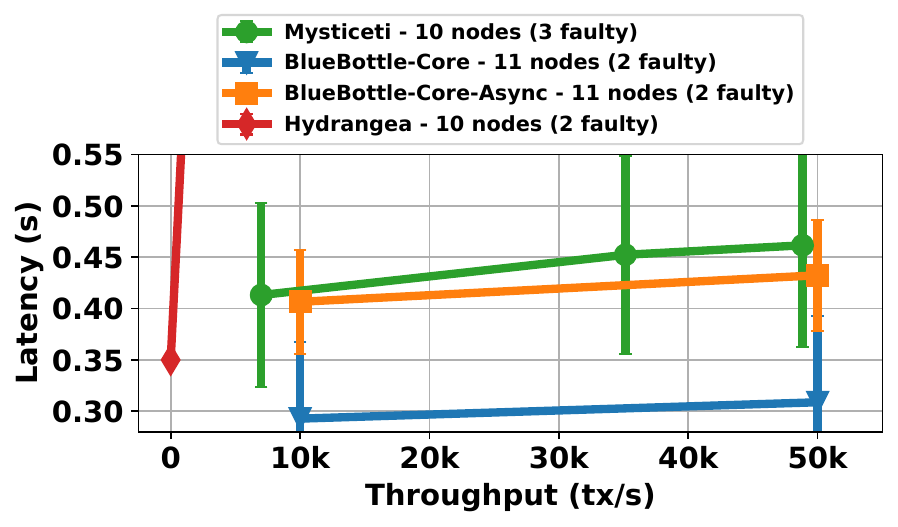}
        \caption{With $3$ or $2$ crash faults; committees with $10$ and $11$ validators (minimum to tolerate $3$ or $2$ faults).}
        \label{fig:evaluation-faults-async}
    \end{subfigure}
    \caption{
        WAN throughput-latency performance comparison of \lmysticeti ($n=5f+1$) and \asynclmysticeti ($n=5f+1$), with Mysticeti ($n=3f+1$) and Hydrangea ($n=3f+2c+k+1$) as reference points. Transaction size is $512$\,B. The y-axis starts at $300$\,ms to zoom in on the latency difference between the systems.
    }
    \label{fig:evaluation-async}
\end{figure}
\section{Reproducing Experiments} \label{sec:tutorial}

Below, we provide a tutorial for reproducing the experiments.
All artifacts necessary for evaluating the contributions of this paper are publicly available at:
\begin{center}
  \codelink
\end{center}
We provide the orchestration scripts used to benchmark the \sysname codebase on AWS and produce the benchmarks of \Cref{sec:evaluation}.

\para{Deploying a testbed}
The file `\texttildelow/.aws/credentials' should have the following content:
\begin{scriptsize}\begin{verbatim}
[default]
aws_access_key_id = YOUR_ACCESS_KEY_ID
aws_secret_access_key = YOUR_SECRET_ACCESS_KEY
\end{verbatim}\end{scriptsize}
configured with account-specific AWS \emph{access key id} and \emph{secret access key}. It is advise to not specify any AWS region as the orchestration scripts need to handle multiple regions programmatically.

A file `settings.yaml' contains all the configuration parameters for the testbed deployment. We run the experiments of \Cref{sec:evaluation} with the following settings:

\begin{scriptsize}\begin{lstlisting}[language=yaml]
---
testbed_id: "${USER}-bluebottle-testbed"
cloud_provider: aws
token_file: "/Users/${USER}/.aws/credentials"
ssh_private_key_file: "/Users/${USER}/.ssh/aws"
regions:
  - us-east-1
  - us-west-2
  - ca-central-1
  - eu-central-1
  - eu-west-1
  - eu-west-2
  - eu-west-3
  - eu-north-1
  - ap-south-1
  - ap-southeast-1
  - ap-southeast-2
  - ap-northeast-1
  - ap-northeast-2
specs: m5d.8xlarge
repository:
  url: https://github.com/AUTHOR/REPO.git
  commit: main
node_parameters_path:
  "crates/orchestrator/assets/node-parameters.yml"
client_parameters_path:
  "crates/orchestrator/assets/client-parameters.yml"
benchmark_duration: 1000
\end{lstlisting}\end{scriptsize}

where the file `/Users/\${USER}/.ssh/aws' holds the ssh private key used to access the AWS instances, and `AUTHOR' and `REPO' are respectively the GitHub username and repository name of the codebase to benchmark.
To run benchmarks with faults (\Cref{sec:faults}), we add the following configuration at the end of the `settings.yaml' file:
\begin{scriptsize}\begin{lstlisting}[language=yaml]
faults: !Permanent { faults: 3 }
\end{lstlisting}\end{scriptsize}

The orchestrator binary provides various functionalities for creating, starting, stopping, and destroying instances. For instance, the following command to boots 2 instances per region (if the settings file specifies 13 regions, as shown in the example above, a total of 26 instances will be created):
\begin{scriptsize}\begin{verbatim}
cargo run --bin orchestrator -- testbed deploy --instances 2
\end{verbatim}\end{scriptsize}
The following command displays he current status of the testbed instances
\begin{scriptsize}\begin{verbatim}
cargo run --bin orchestrator testbed status
\end{verbatim}\end{scriptsize}
Instances listed with a green number are available and ready for use and instances listed with a red number are stopped. It is necessary to boot at least one instance per load generator, one instance per validator, and one additional instance for monitoring purposes (see below).
The following commands respectively start and stop instances:
\begin{scriptsize}\begin{verbatim}
cargo run --bin orchestrator -- testbed start
cargo run --bin orchestrator -- testbed stop
\end{verbatim}\end{scriptsize}
It is advised to always stop machines when unused to avoid incurring in unnecessary costs.

\para{Running Benchmarks}
Running benchmarks involves installing the specified version of the codebase on all remote machines and running one validator and one load generator per instance. For example, the following command benchmarks a committee of 50 validators under a constant load of 1,000 tx/s:
\begin{scriptsize}\begin{verbatim}
cargo run --bin orchestrator -- benchmark \
    --committee 50 fixed-load --loads 1000
\end{verbatim}\end{scriptsize}

The nodes and clients configuration files (respectively specified in `crates/orchestrator/assets/node-parameters.yml' and then `crates/orchestrator/assets/client-parameters.yml') are used to respectively parametrize the nodes and clients. We run our benchmarks with the nodes configuration file as follows:
\begin{scriptsize}\begin{lstlisting}[language=yaml]
leader_timeout:
  secs: 1
  nanos: 0
\end{lstlisting}\end{scriptsize}
and the client configuration file as follows:
\begin{scriptsize}\begin{lstlisting}[language=yaml]
initial_delay:
  secs: 400
  nanos: 0
\end{lstlisting}\end{scriptsize}

\para{Monitoring}
The orchestrator provides facilities to monitor metrics. It deploys a Prometheus instance and a Grafana instance on a dedicated remote machine. Grafana is then available on the address printed on \texttt{stdout} when running benchmarks with the default username and password both set to \texttt{admin}. An example Grafana dashboard can be found in the file `grafana-dashboard.json'\footnote{\dashboardlink}.

\fi

\end{document}